\documentclass[aps,prl,reprint,superscriptaddress,floatfix,footnote, onecolumn,notitlepage]{revtex4-1}

\usepackage{appendix}

\usepackage{tikz}  
\usetikzlibrary{arrows,shapes,positioning,shadows,backgrounds,fit}

\usepackage{relsize}
\usepackage{graphicx}
\usepackage{amsmath, bbm}
\usepackage{amssymb}
\usepackage{amsthm}
\usepackage{mathrsfs}
\usepackage{xcolor}
\usepackage{cancel}
\usepackage{titlesec}
\usepackage{enumitem}  
\usepackage{mathtools}
\usepackage{subfig}

\usepackage[normalem]{ulem}

\titlespacing{\section}{0ex}{2ex}{0.4ex}
\titleformat{name=\section}{\bf}{\thesection.}{.3em}{}

\usepackage{soul}
\usepackage{dsfont}

\def\be{\begin{eqnarray}}
\def\ee{\end{eqnarray}}
\newcommand{\tr}[1]{\text{tr}\left(#1\right)}
\newcommand{\Tr}[1]{\text{Tr}\left[#1\right]}
\newcommand{\<}{\langle}

\newcommand{\ket}[1]{|{#1}\rangle}
\newcommand{\bra}[1]{\langle{#1}|}

\theoremstyle{plain}

\newtheorem{cor}{Corollary}
\newtheorem{thm}{Theorem}
\newtheorem{conj}{Conjecture}

\newcommand{\RE}{\mathbb{R}}

\newcommand{\Nn}{\mathcal{M}_{N}}
\newcommand{\TAN}{\mathcal{T}_{(\Vb)}\Nn}

\newcommand{\id}{\mathbb{I}}
\newcommand{\Vb}{\mathbf{V}}
\newcommand{\Xb}{\mathbf{X}}
\newcommand{\Ab}{\mathbf{A}}
\newcommand{\Ub}{\mathbf{U}}

\newcommand{\Bb}{\mathbf{B}}
\newcommand{\Cb}{\mathbf{C}}
\newcommand{\Db}{\mathbf{D}}
\newcommand{\Hb}{\mathbf{H}}
\newcommand{\Sb}{\mathbf{S}}
\newcommand{\Yb}{\mathbf{Y}}

\newcommand{\bb}{\mathbf{b}}
\newcommand{\eb}{\mathbf{e}}

\newcommand{\Omb}{\mathbf{\Omega}}

\definecolor{myblue}{rgb}{0.2,0.2,0.8}
\definecolor{myblack}{rgb}{0,0,0}
\definecolor{myurl}{rgb}{0.1,0.1,0.4}

\usepackage[colorlinks=true,citecolor=myblue,linkcolor=myblack,urlcolor=myurl]{hyperref}

\begin{document}

\title{Curvature of Gaussian quantum states}

\date{\today}

\author{Harry J.~D. Miller}
\affiliation{Department of Physics and Astronomy, The University of Manchester, Manchester M13 9PL, UK}

\begin{abstract}
    The space of quantum states can be endowed with a metric structure using the second order derivatives of the relative entropy, giving rise to the so-called Kubo-Mori-Bogoliubov inner product. We explore its geometric properties on the submanifold of faithful, zero-displacement Gaussian states parameterised by their covariance matrices, deriving expressions for the geodesic equations, curvature tensors and scalar curvature. Our analysis suggests that the curvature of the manifold is strictly monotonic with respect to the von Neumann entropy, and thus can be interpreted as a measure of state uncertainty. This provides supporting evidence for the Petz conjecture in continuous variable systems. 
\end{abstract}

\maketitle

\section{Introduction}

\

Information geometry is a multidisciplinary branch of statistics used for analysing the space of parametric probability distributions by introducing notions of distance and curvature \cite{amari2016information,nielsen2020elementary}. The field has found applications in a diverse number of areas, such as statistical physics \cite{ruppeiner1998riemannian,janke2004information}, complexity theory \cite{felice2018information}, signal processing \cite{gambini2014parameter}, inductive inference \cite{caticha2004maximum}, and machine learning \cite{amari2010information}. It is also closely tied to geometric approaches to thermodynamic fluctuation theory \cite{ruppeiner1995riemannian}, providing tools for minimising entropy production \cite{salamon1983thermodynamic} and even investigating black-hole thermodynamics \cite{mirza2007ruppeiner,ruppeiner2008thermodynamic}. The concepts of information geometry can also be  extended to describe quantum mechanical systems, where it has arguably found even more fruitful applications including the foundations of quantum theory \cite{brody2001geometric}, metrology \cite{braunstein1994statistical,sidhu2020geometric}, quantum thermodynamics \cite{scandi2019thermodynamic,brandner2020thermodynamic}, phase transitions \cite{banchi2014quantum}, random state preparation \cite{hall1998random,zyczkowski2011generating} and quantum speed limits \cite{pires2016generalized,deffner2017quantum}. 

The formalism of quantum information geometry begins by equipping the space of parameterised quantum states with a Riemannian metric constructed from the second order derivatives of a chosen divergence measure \cite{amari2016information,bengtsson2017geometry}. This gives rise to a notion of \textit{statistical distance} between different states, where length is assigned based upon their degree of distinguishability \cite{wootters1981statistical}. One choice of measure is the relative entropy which induces the so-called Kubo-Mori-Bogoliubov (KMB) inner product on the state manifold \cite{petz1993bogoliubov,petz1994geometry,balian2014entropy}. The KMB metric is notable because it naturally arises in linear response theory \cite{kubo1966fluctuation} and quantum estimation theory \cite{hayashi2002two}. It has been used to define quantum generalisations of thermodynamic length \cite{abiuso2020geometric} and address foundational problems in statistical mechanics \cite{balian1986dissipation,floerchinger2020thermodynamics}. One of the most interesting properties of the KMB metric is the form of its the scalar curvature \cite{michor2000curvature,dittmann2000curvature,gibilisco2005monotonicity}, which is believed to share a deep connection with entropy \cite{petz2002covariance}. A conjecture due to Petz states that the KMB scalar curvature monotonically increases as a state becomes more mixed \cite{petz1994geometry}, implying that it should be viewed as an entropic-like measure of statistical uncertainty. This information-theoretic interpretation of the curvature remains unproven but has been confirmed in a variety of finite-dimensional systems \cite{petz1994geometry,dittmann2000curvature}. In thermodynamic systems, the KMB curvature is also connected to the strength and direction of interparticle interactions \cite{janyszek1986geometrical} and can be used as a witness of criticality \cite{janyszek1989riemannian}. Overall, there is motivation for studying the KMB curvature as a theoretical tool for understanding the geometric structure of quantum mechanical state space alongside its relationship to information theory and statistical mechanics.

While the KMB metric has been well studied in finite-dimensional quantum systems, almost nothing is known about its structure within continuous variable systems. Here an important class are the Gaussian quantum states which commonly arise in quantum optical settings \cite{wang2007quantum,weedbrook2012gaussian,adesso2014continuous}. Gaussian states are unique due to the symplectic structure of their state space \cite{simon1987gaussian}, and so there has been a considerable amount of attention devoted to understanding their geometric properties \cite{tanaka2006kubo,link2015geometry,marian2016quantum,felice2017volume,di2020complexity,siudzinska2019geometry,windt2021local,mehboudi2022thermodynamic,poggi2021diverging}. In order to facilitate an understanding of the KMB metric in infinite-dimensional settings, this paper provides a full characterisation of its properties within the manifold of faithful $N$-mode, zero-displacement Gaussian states, including derivations of the geodesic equation, Riemann and Ricci curvature tensors, and finally the scalar curvature. It is shown that the KMB metric is invariant under symplectic transformations; this important property is used to obtain a formula for the scalar curvature purely in terms of the symplectic eigenvalues of the Gaussian state covariance matrix, allowing for direct calculations of the curvature with arbitrarily large numbers of modes. Using the formula we are able to prove the Petz conjecture for the single mode case. Secondly, we provide both analytic and numerical supporting evidence that the Petz conjecture holds true in larger Gaussian systems, as we observe monotonic growth in the scalar curvature with increasing temperature in a periodic Gaussian thermal chain. In all cases the curvature is found to be negative, recovering the classical constant-curvature prediction in the highly mixed limit while diverging whenever any of the modes are pure. 

\ 

\section{Information geometry and curvature of classical Gaussian distributions}

\

Before we begin discussing Gaussian quantum states it is useful to first recall the well-known information geometry underlying the classical Gaussian distributions. A comprehensive overview of the Gaussian family geometry can be found in \cite{andai2009geometry}. Given a random vector $\underline{x}\in\mathbb{R}^d$, the family of $d$-dimensional standard Gaussian distributions take the form
\begin{align}\label{eq:gauss}
    p(\underline{x};\Vb):=\sqrt{\frac{1}{\text{det}(\Vb)}} \ \text{exp}\bigg[-\frac{1}{2}\underline{x}^T \Vb^{-1} \underline{x}\bigg],
\end{align}
where $[\Vb]_{ij}:=\langle x_i x_j \rangle$ represents the positive, symmetric covariance matrix of dimension $d\times d$ and the mean is assumed to be zero ($\langle \underline{x}\rangle=0$). This is a parametric family of distributions with $\Vb$ providing a set of $d(d+1)/2$ independent real variables for the manifold. The space of classical probability distributions is naturally equipped with a Riemannian structure based on the Fisher-Rao metric \cite{amari2016information}, which is found from the second order expansion of the Kullback-Liebler divergence $D[\Vb|\Vb']=\langle \text{ln} \ [p(\underline{x};\Vb)/p(\underline{x};\Vb')]\rangle$ for infinitesimally close distributions. In short, distances between neighbouring distributions can be quantified in terms of their statistical distinguishability. In the case of the standard Gaussian family the Fisher-Rao metric takes the form \cite{burbea1984informative}
\begin{align}\label{eq:siegel}
    ds^2=\frac{1}{2}\Tr{(\Vb^{-1} d\Vb)^2 }.
\end{align}
In geometry this is equivalent to the Siegel metric \cite{siegel2014symplectic}, describing a hyperbolic manifold of constant negative curvature. Moreover, the scalar curvature of the manifold is given by the elegant formula \cite{andai2009geometry}
\begin{align}\label{eq:classcurv}
    \text{Scal}(\Vb)=-\frac{d (d-1)(d+2)}{4},
\end{align}
Therefore, the standard Gaussian family holds a notable place in the field of information geometry due to its simple structure, computable geodesics and scale invariance. Moreover, information geometric analysis of the Gaussian family has found applications in diffusion tensor imaging \cite{lenglet2006statistics} and Gaussian mixture modelling \cite{pinele2020fisher}. We will eventually see that quantum mechanical Gaussian states do not have this simple structure, since the introduction of quantum fluctuations changes the geometric features of the manifold. In particular it will be shown that the curvature is no longer constant, signifying a breakdown of scale invariance and classicality.

\section{Quantum information geometry and Gaussian states}

\

Our goal will be to develop a quantum generalisation of the classical formalism shown above by replacing the standard Gaussian distributions with Gaussian \textit{states}. Let's first consider the geometry of the full space of quantum states. For a separable Hilbert space $\mathcal{H}$ the set of quantum states are the positive, unit-trace density operators $\mathcal{S}(\mathcal{H}):=\{\hat{\rho} \ | \ \hat{\rho}\geq 0, \ \tr{\hat{\rho}}=1 \}$. An important quantifier of distinguishability between two states is given by the quantum relative entropy,
\begin{align}\label{eq:rel_ent}
    S(\hat{\rho}|| \hat{\rho}')=\tr{\hat{\rho}(\text{ln} \ \hat{\rho}-\text{ln} \   \hat{\rho}')}
\end{align}
which is the non-commutative generalisation of the Kullback-Liebler divergence \cite{umegaki1962conditional}. Following the same rationale used in classical information geometry, one can build a metric connecting different states from the partial derivatives of the relative entropy \cite{petz1994geometry,lesniewski1999monotone}. This can be characterised by the squared line element $ds^2=g_{\hat{\rho}}(d\hat{\rho},d\hat{\rho})$ where
\begin{align}\label{eq:KMB}
    g_{\hat{\rho}}(\hat{A},\hat{B}):=-\frac{\partial^2}{\partial s\partial t }\bigg|_{s=t=0} S\big(\hat{\rho}+t\hat{A}|| \ \hat{\rho}+s\hat{B}\big),
\end{align}
The smooth bilinear form $g_{\hat{\rho}}(\hat{A},\hat{B})$ is symmetric and positive definite, thus defining a Riemannian metric for $\mathcal{S}(\mathcal{H})$. In the literature this is often referred to as the Kubo-Mori-Bogoliubov (KMB) inner product \cite{petz1993bogoliubov}, and it belongs to the family of monotone metrics which all provide non-commutative generalisations of the classical Fisher-Rao metric \cite{petz1996monotone}. It is notable that the KMB metric can alternatively be obtained from the Hessian of the von Neumann entropy $S(\hat{\rho})=-\tr{\hat{\rho} \ \text{ln} \ \hat{\rho}}$ since $ds^2=-d^2 S(\hat{\rho})$ \cite{balian2014entropy}.

In this paper we will be interested in the geometric structure of the subset of Gaussian quantum states, and so we begin by considering the set of continuous-variable quantum systems represented by an $N$-mode bosonic state, where $N$ is a fixed positive integer. The Hilbert space of such  states are given by $\mathcal{H}^{\otimes N}=\bigotimes_{k=1}^N \mathcal{H}_k$, where $\mathcal{H}_k$ represents an infinite dimensional and separable Hilbert space. Let $\hat{x}_j=(\hat{q}_j,\hat{p}_j)$ denote a quadrature operator for the subspace $\mathcal{H}_k$, with $\hat{q}_j$ denoting position and $\hat{p}_j$ the momentum. One then has $2N$ quadrature operators in total, composing a vector $\underline{\hat{x}}=[\hat{q}_1,...,\hat{q}_N, \hat{p}_1,...,\hat{p}_N]=[\hat{x}_1,...,\hat{x}_N]$. These quadrature operators fulfill the commutation relations
\begin{align}
    [\hat{x}_j,\hat{x}_k]=i\Omega_{jk},
\end{align}
where $\Omega_{jk}$ is an element of the $2N\times 2N$ matrix
\begin{align}
    \Omb:=\bigoplus^{N}_{k=1} \ \mathbf{\omega}, \ \ \ \ \ \ \  \mathbf{\omega}=\left(\begin{array}{cc}
          0 & 1   \\ 
          -1 & 0  
    \end{array}\right)
\end{align}
known as the symplectic form (we take $\hbar=1$ throughout). The subset of \textit{standard} Gaussian states $\{\hat{\rho}(\Vb)\}\subset \mathcal{S}(\mathcal{H}^{\otimes N})$ take the form \cite{banchi2015quantum}
\begin{align}\label{eq:gauss}
    \hat{\rho}(\Vb):=\sqrt{\frac{1}{\text{det} \ (\Vb+i\Omb/2)}} \ \text{exp}\bigg[-\frac{1}{2}\underline{\hat{x}}^T \Hb_\Vb \underline{\hat{x}}\bigg]
\end{align}
These states are parameterised entirely in terms of a $2N\times 2N$-dimensional symmetric covariance matrix $\Vb$, which has elements
\begin{align}
    \big[\Vb\big]_{jk}:=\frac{1}{2}\big<\big\{\hat{x}_j, \ \hat{x}_k\big\}_+\big>_{\hat{\rho}}
\end{align}
Note that in analogy to the standard family of Gaussian probability distributions~\eqref{eq:gauss}, we define the the class of standard Gaussian states as those with zero displacement, $\langle \underline{\hat{x}} \rangle_{\hat{\rho}}=0$. Furthermore, the matrix $\Hb_\Vb\geq0$ is a $2N\times 2N$ positive-definite matrix determined by $\Vb$ according to
\begin{align}\label{eq:cayley}
    \Hb_\Vb:=2i\Omb \ \text{arccoth}\big(2\Vb i\Omb\big).
\end{align}
This function may be evaluated using \textit{Williamson's theorem} \cite{weedbrook2012gaussian}, which states that there exists a diagonalisation of the covariance matrix such that 
\begin{align}\label{eq:symp}
    \Vb=\Sb_\Vb \big(\Db_\Vb\oplus\Db_\Vb\big)\Sb_\Vb^T,
\end{align}
where $\Sb_\Vb$ is a symplectic matrix (ie. it preserves $\Sb_\Vb \ \Omb \  \Sb_\Vb^T=\Omb$), and $\Db_\Vb=\text{diag}(\nu_1,...\nu_N)$ is a diagonal matrix of symplectic eigenvalues. Crucially, a valid covariance matrix state must satisfy the uncertainty principle \cite{simon1994quantum}
\begin{align}\label{eq:uncertain}
    \Vb+i\Omb/2\geq 0. 
\end{align}
This is saturated by the pure vacuum state $\ket{0}\bra{0}$ which has minimal uncertainty in position and momentum. As a consequence of~\eqref{eq:uncertain} the symplectic eigenvalues must be bounded as $\nu_j\geq 1/2 \ \forall j\in\{1,..,N\}$. A Gaussian state is said to be \textit{faithful} if and only if $\nu_j> 1/2 \ \forall j\in\{1,..,N\}$ \cite{seshadreesan2018renyi}; in this paper we will always assume this to be true and only consider non-faithful states as limiting cases of all subsequent expressions. The uncertainty relation provides a clear distinction between the family of Gaussian states versus Gaussian probability distributions, since~\eqref{eq:uncertain} shows that not all positive covariance matrices are consistent with a valid quantum state. The relative entropy between two Gaussian states parameterised by $\Vb$ and $\Vb'$ respectively is given by \cite{chen2005gaussian}
\begin{align}\label{eq:rel_ent_gauss}
    S\big(\hat{\rho}(\Vb)|| \hat{\rho}(\Vb')\big)=\frac{1}{4}\ln \big[\mathcal{Z}(\Vb')/\mathcal{Z}(\Vb)\big]+\frac{1}{2}\Tr{(\Hb_{\Vb'}-\Hb_{\Vb})\Vb}.
\end{align}

In analogy to the classical case we may view the set of covariance matrices as a collection of parameters on a manifold. Let $\Nn$ denote the set of real $2N\times 2N$ symmetric matrices that fulfill the uncertainty relation~\eqref{eq:uncertain} in the strict sense, ie. $\nu_j> 1/2 \ \forall j\in\{1,..,N\}$. The set of faithful standard Gaussian quantum states then form a manifold parameterised by the matrices $\Vb\in\Nn$ and we denote the corresponding tangent space by $\mathcal{T}_{(\Vb)}\Nn$. To identify this tangent space, consider the directional derivative at a point $\Vb\in\Nn$ along the vector $\Ab\in\mathcal{T}_{(\Vb)}\Nn$, defined for a smooth function $h:\Nn\to \RE$ by
\begin{align}\label{eq:deriv}
    \frac{\partial h(\Vb)}{\partial\Ab}:=\frac{d}{dt}h(\Vb+t\Ab)\bigg|_{t=0},
\end{align}
From this we can see that the tangent space is equivalent to the set of $2N\times 2N$ real, symmetric matrices. The KMB metric~\eqref{eq:KMB} on the Gaussian state manifold is then a smooth map of the form
\begin{align}
    g_{\Vb}: \ \TAN\times\TAN \to \RE \ \ \ \bigg(\big(\Ab,\Bb)\mapsto g_{\Vb}\big(\Ab,\Bb\big)\bigg),
\end{align}
for any $\Vb\in\Nn$, which is given by the negative Hessian of the quantum relative entropy~\eqref{eq:rel_ent_gauss} with respect to variations in the covariance matrix,
\begin{align}\label{eq:relentmet}
    g_{\Vb}(\Ab,\Bb):=-\frac{1}{2}\frac{\partial^2}{\partial s\partial t }\bigg|_{s=t=0} \Tr{\Hb_{\Vb+s\Bb}(\Vb+t\Ab)},
\end{align}
It will be useful throughout our analysis to note the following matrix derivative  which will be used repeatedly:
\begin{align}
    \frac{d}{dt}[\Xb(t)]^{-1}=-[\Xb(t)]^{-1} \dot{\Xb}(t)[\Xb(t)]^{-1}, 
\end{align}
To begin we first introduce an integral representation of the matrix $\Hb_\Vb$, which is proven in Appendix A:
\begin{align}\label{eq:int_rep}
    \Hb_\Vb=\int^1_{-1}d\lambda \ \big[2\Vb+\lambda(i\Omb)\big]^{-1}.
\end{align}
It then follows that the directional derivative is
\begin{align}\label{eq:deriv_D}
    \frac{d}{dt}\Hb_{\Vb+t\Cb}\bigg|_{t=0}=- 2\Delta_\Vb[\Cb]
\end{align}
where
\begin{align}\label{eq:delta_map}
    \Delta_\Vb(.):=\int^1_{-1}d\lambda \ \big[2\Vb+i\lambda\Omb\big]^{-1}(.)  \ \big[2\Vb+i\lambda\Omb\big]^{-1}.
\end{align}
We combine~\eqref{eq:deriv_D} with~\eqref{eq:relentmet} to get the metric:
\begin{align}\label{eq:metric_final}
    g_{\Vb}\big(\Ab,\Bb\big)=\Tr{\Bb \ \Delta_\Vb(\Ab)},
\end{align}
To verify this is a valid metric for $\Nn$ we need to check $(i)$ symmetry and $(ii)$ positivity. Cyclicity of the trace confirms symmetry in $\Ab,\Bb$, while positivity can be shown by taking the transpose of the uncertainty relation~\eqref{eq:uncertain} to get $\Vb+i\lambda \Omb>0$ for $\lambda\in[-1,1]$. This implies that we can define the inverse hermitian square root $(\Vb+i\lambda \Omb)^{-1/2}$ at each point inside the integral over $d\lambda$, and use cyclicity of the trace to get
\begin{align}
    g_{\Vb}\big(\Ab,\Ab\big)=\int^1_{-1}d\lambda \ \Tr{\bigg(2\Vb+i\lambda \Omb)^{-1/2}\Ab (2\Vb+i\lambda \Omb)^{-1/2}\bigg)^\dagger\bigg(2\Vb+i\lambda \Omb)^{-1/2}\Ab (2\Vb+i\lambda \Omb)^{-1/2}\bigg)}\geq  0,
\end{align}
as desired. The pair $(\Nn,g)$ thus defines a Riemann manifold for the faithful, standard  $N-$mode Gaussian states, and when choosing the covariance matrices as the parameters the corresponding squared line element is
\begin{align}\label{eq:ds}
    ds^2=\int^1_{-1}d\lambda \ \Tr{d\Vb\big[2\Vb+i\lambda\Omb\big]^{-1}d\Vb  \ \big[2\Vb+i\lambda\Omb\big]^{-1}},
\end{align}
which is a quantum generalisation of~\eqref{eq:siegel}. We now state an important symmetry  property of the metric,
\begin{thm}
    The KMB metric on the manifold of standard Gaussian states is invariant under symplectic transformations  of the form
    \begin{align}
        g_{\Sb\Vb \Sb^T}\big(\Sb\Ab\Sb^T,\Sb\Bb \Sb^T\big)=g_{\Vb}\big(\Ab,\Bb\big).
    \end{align}
where $\Sb$ is any symplectic matrix.     
\end{thm}
\begin{proof}
    Since $\Sb\Omb \Sb^T=\Omb$ we may write
    \begin{align}\label{eq:Sinv}
        [2\Sb \Vb \Sb^T+i\lambda\Omb]^{-1}=[\Sb(2\Vb +i\lambda\Omb)\Sb^T]^{-1}=(\Sb^T)^{-1}[2\Vb +i\lambda\Omb]^{-1}\Sb^{-1}.
    \end{align}
By cyclicity of the trace we have
\begin{align}
    \nonumber g_{\Sb\Vb \Sb^T}\big(\Sb\Ab\Sb^T,\Sb\Bb \Sb^T\big)&=\int^1_{-1}d\lambda \ \Tr{\Sb^{-1}\Sb\Bb\Sb^T (\Sb^T)^{-1}\big[2\Vb+i\lambda\Omb\big]^{-1}\Sb^{-1}\Sb \Ab \Sb^T (\Sb^T)^{-1}  \ \big[2\Vb+i\lambda\Omb\big]^{-1}}, \\
    &=g_{\Vb}\big(\Ab,\Bb\big).
\end{align}
\end{proof}
The symplectic invariance of the metric is reminiscent  of the general unitary invariance of the KMB metric on the full space of quantum states \cite{petz1996monotone}, since any unitary transformation of the form $\hat{\rho}\mapsto \hat{U}\hat{\rho}\hat{U}^\dagger$ that connects two Gaussian states is equivalent to applying a symplectic transformation on the corresponding covariance matrix $\Vb\mapsto \Sb\Vb\Sb^T$ \cite{simon1987gaussian}.

To explore the geometric properties of the KMB metric we can introduce a covariant derivative arising from the Levi-Civita connection of $g$ \cite{amari2016information,nielsen2020elementary}. One can start by taking the directional derivative of the metric, given by
\begin{align}\label{eq:metric_deriv}
    dg_{\Vb}(\Cb)\big(\Ab,\Bb\big):=\frac{\partial }{\partial\Cb} \ g_{\Vb}\big(\Ab,\Bb\big)\big)
\end{align}
For a given point $\Vb\in\Nn$ and tangent vectors $\Ab,\Bb\in\TAN$, define the linear functional
\begin{align}\label{eq:theta}
    \Theta_{(\Vb,\Ab,\Bb)}: \ \TAN\to\RE \ \ \ \ \Cb\mapsto \frac{1}{2}\bigg(dg_{\Vb}(\Bb)\big(\Ab,\Cb\big)+dg_{\Vb}(\Ab)\big(\Bb,\Cb\big)-dg_{\Vb}(\Cb)\big(\Ab,\Bb\big)\bigg),
\end{align}
Linearity with respect to $\Cb$ implies that for any $\Cb\in\TAN$ we can find a unique tangent vector $\Gamma_{\Vb}\big(\Ab,\Bb\big)\in\TAN$, such that
\begin{align}\label{eq:covderiv}
    g_{\Vb}\big(\Gamma_{\Vb}\big(\Ab,\Bb\big),\Cb\big)=\Theta_{(\Vb,\Ab,\Bb)}(\Cb).
\end{align}
We define this vector as the \textit{covariant derivative}. To find its explicit form, consider the map
\begin{align}\label{eq:deriv_g2}
    d\Delta_{\Vb}(\Cb)(\Ab):= -\frac{\partial }{\partial\Cb}\Delta_{\Vb}(\Ab)=2\int^1_{-1}d\lambda \ \theta_\lambda(\Vb)\bigg(\Cb \ \theta_\lambda(\Vb) \ \Ab+\Ab \ \theta_\lambda(\Vb) \ \Cb\bigg)\theta_\lambda(\Vb)
\end{align}
where for shorthand we denote the matrix function $\theta_\lambda(\Vb):=[2\Vb+i\lambda\Omb]^{-1}$. The metric derivative is then 
\begin{align}
    dg_{\Vb}(\Cb)\big(\Ab,\Bb\big)=-\Tr{\Bb \  d\Delta_{\Vb}(\Cb)(\Ab)},
\end{align}
Since this is symmetric with respect to $\Ab,\Bb,\Cb$ due to cyclicity of the trace, we have from~\eqref{eq:theta}
\begin{align}\label{eq:theta1}
    \Theta_{(\Vb,\Ab,\Bb)}(\Cb)=-\frac{1}{2}\Tr{\Cb \  d\Delta_{\Vb}(\Bb)(\Ab)}
\end{align}
Finally, we note that the map $\Delta_{\Vb}(.)$ has a well-defined inverse $\Delta_{\Vb}^{-1}(.)$ (an explicit expression is given in the next section). Comparing~\eqref{eq:theta1} with~\eqref{eq:covderiv} verifies
\begin{align}\label{eq:covderiv_test}
    \Gamma_{\Vb}\big(\Ab,\Bb\big)=-\frac{1}{2}\Delta_{\Vb}^{-1}\cdot d\Delta_{\Vb}(\Bb)(\Ab),
\end{align}
The metric and its covariant derivative are enough to characterise the geodesic paths on $\Nn$. Recalling the geodesic equation, we can say that a curve $\gamma:\mathbb{R}\mapsto \Nn$, $\gamma(t)=\Vb(t)$, $t\in[0,1]$ is a geodesic if an only if
\begin{align}
    \forall t\in[0,1]: \ \ \ \ \ \ddot{\gamma}(t)=\frac{1}{2}\Delta_{\gamma(t)}^{-1}\cdot d\Delta_{\gamma(t)}\big(\dot{\gamma}(t)\big)\big(\dot{\gamma}(t)\big).
\end{align}
and the geodesic length connecting two points $\Vb_0=\gamma(0)$ and $\Vb_1=\gamma(1)$ is
\begin{align}\label{eq:length}
    \mathcal{L}(\Vb_0,\Vb_1)=\int^1_0 dt \ \sqrt{\Tr{\dot{\gamma}(t) \ \Delta_{\gamma(t)}(\dot{\gamma}(t))}}
\end{align}
Since we have constructed this distance from the relative entropy, $\mathcal{L}(\Vb_0,\Vb_1)$ quantifies the effort at which we can distinguish the two states, playing the role of a statistical distance \cite{wootters1981statistical}. Unlike the classical case, analytic expressions for the geodesic curves of the KMB metric under arbitrary boundary conditions are not yet known but we will not need to pursue this in the present paper.

\section{Scalar curvature of the Gaussian state manifold }

\

\noindent In this section we now move on to determining the curvature properties of $(\Nn,g)$ culminating in a closed form expression for the Ricci scalar curvature. We will assume familiarity with the general concepts and definitions of curvature in Riemannian geometry, though for background the reader can refer to \cite{amari2016information,nielsen2020elementary}. To begin we will need the derivative of the covariant derivative, which is a map
\begin{align}\label{eq:deriv_cov_deriv2}
     d\Gamma_{\Vb}(\Cb)\big(\Ab,\Bb\big)&=\frac{\partial}{\partial\Cb}\Gamma_{\Vb}\big(\Ab,\Bb\big)=\frac{1}{2}\Delta_{\Vb}^{-1}\cdot d\Delta_{\Vb}(\Cb)\cdot\Delta_{\Vb}^{-1}\cdot d\Delta_{\Vb}(\Bb)(\Ab)-\frac{1}{2}\Delta_{\Vb}^{-1}\cdot \frac{\partial}{\partial \Cb} d \Delta_{\Vb}(\Bb)(\Ab).
\end{align}
The Riemann curvature tensor is defined by
\begin{align}
    R_{\Vb}\big(\Ab,\Bb,\Cb\big):=d\Gamma_{\Vb}(\Ab)\big(\Bb,\Cb\big)\big)-d\Gamma_{\Vb}(\Bb)\big(\Ab,\Cb\big)\big)+\Gamma_{\Vb}\big(\Ab,\Gamma_{\Vb}\big(\Bb,\Cb\big)\big)-\Gamma_{\Vb}\big(\Bb,\Gamma_{\Vb}\big(\Ab,\Cb\big)\big),
\end{align}
Combining~\eqref{eq:deriv_cov_deriv2} with~\eqref{eq:covderiv_test} and using the fact that $\frac{\partial}{\partial \Cb} d \Delta_{\Vb}(\Bb)(\Ab)$ is symmetric in $\Ab,\Bb,\Cb$, we find
\begin{align}\label{eq:riemann}
    R_{\Vb}\big(\Ab,\Bb,\Cb\big)=\frac{1}{4}\Delta_{\Vb}^{-1}\cdot d\Delta_{\Vb}(\Bb)\cdot\Delta_{\Vb}^{-1}\cdot d\Delta_{\Vb}(\Ab)(\Cb)-\frac{1}{4}\Delta_{\Vb}^{-1}\cdot d\Delta_{\Vb}(\Ab)\cdot\Delta_{\Vb}^{-1}\cdot d\Delta_{\Vb}(\Bb)(\Cb).
\end{align}
Alongside this we can also introduce the Ricci curvature tensor, which is defined from the trace of the map $R_{\Vb}\big((.),\Bb,\Cb\big)$,
\begin{align}
    \text{Ric}_{\Vb}\big(\Ab,\Bb\big):=\Tr{R_{\Vb}\big((.),\Ab,\Bb\big)}.
\end{align}
Now let the set $\{\Xb_s\}_{s=1,2,...2N^2+N}$ be an orthonormal basis for the tangent space $\TAN$. The Ricci curvature tensor can then be determined by expanding within this chosen basis, so that 
\begin{align}\label{eq:ric}
\text{Ric}_{\Vb}\big(\Ab,\Bb\big)=\sum_{s=1}^{2N^2+1}\Tr{R_{\Vb}\big(\Xb_s,\Ab,\Bb\big)\Xb_s}.
\end{align}
For this linear map we can always find a unique tangent vector $\tilde{\Ab}\in\TAN$ such that
\begin{align}
    g_{\Vb}\big(\tilde{\Ab},\Bb\big)=\text{Ric}_{\Vb}\big(\Ab,\Bb\big).
\end{align}
We can therefore define a map 
\begin{align}
    \tilde{\text{Ric}}: \Nn\to\text{Lin}(\TAN,\TAN) \ \ \ \Vb\mapsto\big(\Ab\mapsto\tilde{\Ab}\big),
\end{align}
The KMB scalar curvature of the Gaussian manifold is then given by the trace of this map, namely
\begin{align}\label{eq:scalar}
    \text{Scal}:\Nn\to\mathbb{R} \ \ \ \Vb\mapsto \Tr{\tilde{\text{Ric}}_{\Vb}}.
\end{align}
It is straightforward to show that the scalar curvature at a point $\Vb\in\Nn$ can be expressed as a summation over the basis $\{\Xb_s\}_{s=1,2,...2N^2+N}$, with
\begin{align}\label{eq:scalar1}
    \text{Scal}(\Vb)=\sum^{2N^2+N}_{s,t=1} \big\langle R_{\Vb}\big(\Xb_t,\Xb_s,\Delta^{-1}_{\Vb}(\Xb_t)\big) ,\Xb_s\big\rangle, 
\end{align}
where $\langle \Xb_s, \Xb_t  \rangle=\Tr{\Xb_s^T \Xb_t}$ denotes the Hilbert-Schmidt inner product for real matrices. To perform this complicated summation for an arbitrary number of modes, we can exploit the symplectic structure of the set of Gaussian states and the symplectic invariance of the metric (Theorem 1). We summarise our main result as the following theorem: 
\begin{thm}
    The KMB scalar curvature of an $N$-mode standard Gaussian state with covariance matrix $\Vb$ is a symplectic invariant, and can be expressed in terms of three symmetric functions of the symplectic eigenvalues $\{\nu_1,\nu_2,...\nu_N\}$ of $\Vb$:
\begin{align}\label{eq:mainresult}
    \mathrm{Scal}(\Vb)=\sum_{i=1}^N \varphi_1[\nu_i]+\sum^N_{i<j} \varphi_2[\nu_i,\nu_j]+\sum^N_{i<j<k} \varphi_3[\nu_i,\nu_j,\nu_k],
\end{align}
where
\begin{align}
    \nonumber\varphi_1[\nu_i]:=\frac{B_{iii} (2 A_{iii}g_{ii}-B_{iii} f_{ii})}{f_{ii}^2 \ g_{ii}^2},
\end{align}
\begin{align}
    \nonumber\varphi_2[\nu_i,\nu_j]&:=\frac{A_{iij}}{f_{ii}f_{ij}}\bigg[\frac{A_{iij}}{4f_{ij}}-\frac{B_{iji}}{g_{ij}}-\frac{2B_{iii}}{g_{ii}}-\frac{A_{iii}}{f_{ii}}\bigg]+\frac{A_{ijj}}{f_{jj}f_{ij}}\bigg[\frac{A_{ijj}}{4f_{ij}}-\frac{B_{jij}}{g_{ij}}-\frac{2B_{jjj}}{g_{jj}}-\frac{A_{jjj}}{f_{jj}}\bigg] \\
    \nonumber& \ \ \ \ \ \ \ \ \ +\frac{3}{f_{ij}g_{ij}}\bigg[\frac{B_{iij}^2}{g_{ii}}+\frac{B_{ijj}^2}{g_{jj}}\bigg]+\frac{B_{iji}}{f_{ii}g_{ij}}\bigg[\frac{B_{iji}}{4g_{ij}}-\frac{2B_{iii}}{g_{ii}}-\frac{A_{iii}}{f_{ii}}\bigg]+\frac{B_{jij}}{f_{jj}g_{ij}}\bigg[\frac{B_{jij}}{4g_{ij}}-\frac{2B_{jjj}}{g_{jj}}-\frac{A_{jjj}}{f_{jj}}\bigg]
\end{align}    
\begin{align}
    \nonumber\varphi_3[\nu_i,\nu_j,\nu_k]&:=\frac{3}{2}\bigg[\frac{A_{ijk}^2}{f_{ij}f_{ik}f_{jk}}+\frac{B_{ikj}^2}{f_{ij}g_{ik}g_{jk}}+\frac{B_{ijk}^2}{f_{ik}g_{jk}g_{ij}}+\frac{B_{jik}^2}{f_{jk}g_{ij}g_{ik}}\bigg]-\bigg[\frac{A_{jjk}}{f_{jj}f_{jk}}+\frac{B_{jkj}}{f_{jj}g_{jk}}\bigg]\bigg[\frac{B_{jij}}{g_{ij}}+\frac{A_{ijj}}{f_{ij}}\bigg] \\
    \nonumber& \ \ \ \ \ \ \ \ -\bigg[\frac{A_{ijj}}{f_{ii}f_{ij}}+\frac{B_{iji}}{f_{ii}g_{ij}}\bigg]\bigg[\frac{B_{iki}}{g_{ik}}+\frac{A_{iik}}{f_{ik}}\bigg]-\bigg[\frac{A_{ikk}}{f_{kk}f_{ik}}+\frac{B_{kik}}{f_{kk}g_{ik}}\bigg]\bigg[\frac{B_{kjk}}{g_{jk}}+\frac{A_{jkk}}{f_{jk}}\bigg],
\end{align} 
and 
\begin{align}
\nonumber&f_{ij}:=\frac{f(\nu_j)-f(\nu_i)}{\nu_i-\nu_j}, \ \ \ \ \ \ \ \ g_{ij}:=\frac{f(\nu_i)+f(\nu_j)}{\nu_i+\nu_j}, \\
\nonumber&A_{ijk}:=\bigg[\frac{f_{jk}-f_{ij}}{\nu_i-\nu_k}\bigg], \ \ \ \ \ \ \ \ B_{ijk}:=\bigg[\frac{g_{jk}-g_{ij}}{\nu_i-\nu_k}\bigg]
\end{align}
with $f(x)=\mathrm{arccoth}(2x)$. The formula remains well defined regardless of any degeneracies in the symplectic spectrum.  
\end{thm}
\begin{proof}
The first step is to recognise that the scalar curvature is invariant under symplectic transformations, ie. 
\begin{align}\label{eq:invar}
    \text{Scal}(\Vb)=\text{Scal}(\Sb\Vb\Sb^T),
\end{align}
where $\Sb\Omb\Sb^T=\Omb$ is any symplectic matrix that is independent of the eigenvalue set $\{\nu_1,\nu_2,...\nu_N\}$. This follows from symplectic invariance of the KMB inner product (Theorem 1), implying that the transformation $\Vb\mapsto \Sb \Vb\Sb^T$ is an isometry (ie. preserves the distance~\eqref{eq:length}). The scalar curvature is isometrically invariant \cite{gallot2004riemannian}, hence proving~\eqref{eq:invar}. Next, we can use the diagonalisation~\eqref{eq:symp} of the covariance matrix to show
\begin{align}\label{eq:diag}
    \theta_\lambda(\Vb):=[2\Vb+i\lambda\Omb]^{-1}=\Omb\Sb_\Vb\Omb\tilde{\Ub}\big((2\Db_\Vb-\lambda)^{-1}\oplus (2\Db_\Vb+\lambda)^{-1}\big)\tilde{\Ub}^\dagger\Omb \Sb^T \Omb,
\end{align}
where $\tilde{\Ub}=\Ub\otimes\id_N$ is a unitary matrix with 
\begin{align}
    \Ub=\frac{1}{\sqrt{2}}\left(\begin{array}{cc}
          1 & 1   \\ 
          i & -i  
    \end{array}\right)
\end{align}
and $\Sb$ is a symplectic matrix independent of the eigenvalues as desired (see Appendix B). Then using the invariance~\eqref{eq:invar} we can therefore neglect the symplectic transformation $\Omb \Sb \Omb(.)\Omb\Sb^T\Omb$ and fix
\begin{align}\label{eq:diag2}
    \tilde{\theta}_\lambda(\Vb)=(2\Db_\Vb-\lambda)^{-1}\oplus (2\Db_\Vb+\lambda)^{-1}. 
\end{align}
All subsequent quantities are computed under a replacement $\theta_\lambda(\Vb)\mapsto \tilde{\theta}_\lambda(\Vb)$, and the transformed metric~\eqref{eq:metric_final} becomes \begin{align}
    g_{\Vb}(\Ab,\Bb)=\int^1_{-1}d\lambda \ \Tr{\tilde{\Ub}^\dagger\Bb \tilde{\Ub} \ \tilde{\theta}(\Vb)\tilde{\Ub}^\dagger\Ab\tilde{\Ub} \ \tilde{\theta}(\Vb)},
\end{align}
 Since we have fixed the covariance matrix to be of a diagonal form, it follows from~\eqref{eq:diag2} that we can expand the map $\Delta_{\Vb}(.)$ in terms of the symplectic eigenvalues as follows. Consider the two matrices $\Xb\otimes \bar{\Xb}$ and $\Yb\otimes\bar{\Yb}$, where $\bar{\Xb}$ and $\bar{\Yb}$ are a pair of $N\times N$ matrices while
\begin{align}\label{eq:basis2}
    \Xb=\left(\begin{array}{cc}
          x_1 & x_2   \\ 
          x_3 & x_4  
    \end{array}\right), \ \ \ \  \ \ 
    \Yb=\left(\begin{array}{cc}
          y_1 & y_2   \\ 
          y_3 & y_4  
    \end{array}\right), 
\end{align}
Then we see
\begin{align}
   \nonumber\Delta_{\Vb}(\Xb\otimes \bar{\Xb})&=\int^1_{-1}d\lambda \ \tilde{\theta}(\Vb)(\Xb\otimes \bar{\Xb}) \ \tilde{\theta}(\Vb), \\
   &=\sum_{i,j=1}^{N} \bar{X}_{ij} \ m_{ij}(\Xb)\otimes \eb_{ij},
\end{align}
where $\bar{X}_{ij}$ are the matrix elements of $\bar{\Xb}$ and we have defined the linear map given by the following Hadamard product
\begin{align}
    m_{ij}(\Xb):=\Xb\circ\left(\begin{array}{cc}
          f_{ij} & g_{ij}   \\ 
          g_{ij} & f_{ij}  
    \end{array}\right)
\end{align}
This map is well defined for $\nu_i=\nu_j$ since
\begin{align}
    \lim_{\nu_i\to\nu_j}f_{ij}=\frac{2 }{4\nu_j^2-1};
\end{align}
We can directly determine the inverse
\begin{align}\label{eq:decomp1}
    \Delta_{\Vb}^{-1}(\Xb\otimes \bar{\Xb})=\sum_{i,j=1}^{N} \bar{X}_{ij} \ m^{-1}_{ij}(\Xb)\otimes \eb_{ij},
\end{align}
with
\begin{align}
    m^{-1}_{ij}(\Xb):=\Xb\circ\left(\begin{array}{cc}
          1/f_{ij} & 1/g_{ij}   \\ 
          1/g_{ij} & 1/f_{ij}  
    \end{array}\right)
\end{align}
Similarly we can expand the map $d\Delta_{\Vb}(\Bb)(\Ab)$ as
\begin{align}\label{eq:decomp2}
    \nonumber d\Delta_{\Vb}(\Yb\otimes \bar{\Yb})(\Xb\otimes \bar{\Xb})&=2\int^1_{-1}d\lambda \ \tilde{\theta}_\lambda(\Vb)\bigg((\Xb\otimes \bar{\Xb}) \ \tilde{\theta}_\lambda(\Vb) \ (\Yb\otimes \bar{\Yb})+(\Yb\otimes \bar{\Yb}) \ \tilde{\theta}_\lambda(\Vb) \ (\Xb\otimes \bar{\Xb})\bigg)\tilde{\theta}_\lambda(\Vb), \\
    &=\sum_{i,j,k=1}^{N} m_{ijk}(\Xb,\Yb)\otimes\big(\eb_{ii} \ \bar{\Xb} \ \eb_{jj} \ \bar{\Yb}\eb_{kk}\big)+m_{ijk}(\Yb,\Xb)\otimes\big(\eb_{ii} \ \bar{\Yb} \ \eb_{jj} \ \bar{\Xb} \ \eb_{kk}\big),
\end{align}
where we introduce a quadratic matrix function 
\begin{align}
    m_{ijk}(\Xb,\Yb):=\left(\begin{array}{cc}
          x_1y_1 A_{ijk}+x_2 y_3 B_{ijk} \ \ \   &  \ \ \  x_1 y_2B_{ikj}+x_2y_4 B_{jik}   \\[0.5cm]
          x_4y_3 B_{ikj}+x_3 y_1 B_{jik} \ \ \  &  \ \ \  x_4 y_4 A_{ijk}+x_3 y_2 B_{ijk}  
    \end{array}\right),
\end{align}
The map remains well defined in the presence of degeneracies due to the limits
\begin{align}
    \nonumber&A_{iii}=2 \nu_i f_{ii}^2,  \\
    \nonumber&A_{iji}=A_{iij}, \\
    \nonumber&B_{iii}=\frac{\nu_i f_{ii}+f(\nu_i)}{2\nu_i^2}, \\
    &B_{iji}=\bigg[\frac{g_{ij}+f_{ii}}{\nu_i+\nu_j}\bigg], 
\end{align}
To evaluate the curvature, we will need to choose a basis for the real, symmetric $2N\times 2N$ matrices that is orthonormal with respect to the Hilbert-Schmidt inner product.  It is convenient to choose the following set of $2N^2+N$ matrices:
\begin{align}\label{eq:basis}
    \mathcal{B}_{2N}:=\big\{\mathbf{a}_n\otimes\eb_{jj} \big\}_{\substack{j=1,2,...N \\ n=1,2,3}}\bigcup\big\{\mathbf{a}_n\otimes\bb_{jk}\big\}_{\substack{1\leq j< k\leq N \\ n=1,2}}\bigcup \big\{\mathbf{g}_{jk}\big\}_{1\leq j< k\leq N }\bigcup \big\{\tilde{\mathbf{g}}_{jk}\big\}_{1\leq j< k\leq N}.
\end{align}
with
\begin{align}
    \nonumber&\bb_{jk}=\frac{1}{\sqrt{2}}\big(\eb_{jk}+\eb_{kj}\big), \\
    \nonumber&\tilde{\bb}_{jk}=\frac{1}{\sqrt{2}}\big(\eb_{jk}-\eb_{kj}\big), \\
    \nonumber&\mathbf{g}_{jk}:=\frac{1}{\sqrt{2}}\big(\mathbf{a}_3\otimes\bb_{jk}+\mathbf{a}_4\otimes\tilde{\bb}_{jk}\big), \\
    &\tilde{\mathbf{g}}_{jk}:=\frac{1}{\sqrt{2}}\big(\mathbf{a}_3\otimes\bb_{jk}-\mathbf{a}_4\otimes\tilde{\bb}_{jk}\big),
\end{align}
where $\eb_{jk}$ are the standard matrix units and 
\begin{align}\label{eq:basis2}
    \mathbf{a}_1=\left(\begin{array}{cc}
          1 & 0   \\ 
          0 & 0  
    \end{array}\right), \ \ \ \  \ \ 
    \mathbf{a}_2=\left(\begin{array}{cc}
          0 & 0   \\ 
          0 & 1  
    \end{array}\right), \ \ \ \  \ \ \mathbf{a}_3=\frac{1}{\sqrt{2}}\left(\begin{array}{cc}
          0 & 1   \\ 
          1 & 0  
    \end{array}\right),
    \ \ \ \  \ \ \mathbf{a}_4=\frac{1}{\sqrt{2}}\left(\begin{array}{cc}
          0 & 1   \\ 
          -1 & 0  
    \end{array}\right),
\end{align}
Using this the scalar curvature~\eqref{eq:scalar1} can be written in the form
\begin{align}\label{eq:scalar2}
    \text{Scal}(\Vb)=\sum_{s\neq t} \mathcal{K}(\Xb_s ,\Xb_t ),
\end{align}
where we will set $\{\Xb_s\}=\mathcal{B}_{2N}$ as our basis and define
\begin{align}\label{eq:scalar_HS}
    \nonumber\mathcal{K}(\Ab,\Bb):=\frac{1}{4}&\big\<\Delta_{\Vb}^{-1}\cdot d\Delta_{\Vb}(\tilde{\Ub}^\dagger\Ab \tilde{\Ub} )\cdot\Delta_{\Vb}^{-1}\cdot d\Delta_{\Vb}(\tilde{\Ub}^\dagger\Bb\tilde{\Ub})\cdot \Delta_{\Vb}^{-1}(\tilde{\Ub}^\dagger\Ab\tilde{\Ub}),\tilde{\Ub}^\dagger\Bb\tilde{\Ub} \big> \\
    &-\frac{1}{4}\big\<\Delta_{\Vb}^{-1}\cdot d\Delta_{\Vb}(\tilde{\Ub}^\dagger\Bb\tilde{\Ub} )\cdot\Delta_{\Vb}^{-1}\cdot d\Delta_{\Vb}(\tilde{\Ub}^\dagger\Ab\tilde{\Ub})\cdot \Delta_{\Vb}^{-1}(\tilde{\Ub}^\dagger\Ab \tilde{\Ub}),\tilde{\Ub}^\dagger\Bb \tilde{\Ub} \big>,
\end{align}
To compute this we first expand the following trace functional
\begin{align}
    T(\mathbf{W},\Xb,\Yb,\mathbf{Z})=\Tr{\Delta_{\Vb}^{-1}(\mathbf{W}) d\Delta_{\Vb}(\mathbf{Z} )\cdot\Delta_{\Vb}^{-1}\cdot d\Delta_{\Vb}(\Yb)\cdot \Delta_{\Vb}^{-1}(\Xb)},
\end{align}
where all matrix inputs are of product sum form 
\begin{align}
    \mathbf{W}=\sum_{\nu}\sum_{ij} \bar{w}_{ij}^\nu\mathbf{w}^\nu \otimes \eb_{ij}, \ \ \Xb=\sum_{\nu}\sum_{ij} \bar{x}_{ij}^\nu\mathbf{x}^\nu \otimes \eb_{ij}, \ \ \Yb=\sum_{\nu}\sum_{ij} \bar{y}_{ij}^\nu\mathbf{y}^\nu \otimes \eb_{ij}, \ \ \mathbf{Z}=\sum_{\nu}\sum_{ij} \bar{z}_{ij}^\nu\mathbf{z}^\nu \otimes \eb_{ij}.
\end{align}
with $\{\mathbf{w}^\nu$, $\mathbf{x}^\nu$, $\mathbf{y}^\nu$, $\mathbf{z}^\nu\}$ all $2\times 2$ matrices. Using the expansions~\eqref{eq:decomp1} and~\eqref{eq:decomp2} and contracting indices eventually yields
\begin{align}\label{eq:Trace}
\nonumber T(\mathbf{W},\Xb,\Yb,\mathbf{Z})=&\sum_{\nu\mu\nu'\mu'}\sum_{ijkl}(\bar{w}^{\mu'}_{li} \ \bar{z}^{\nu'}_{kl} \ \bar{y}^\mu_{jk} \ \bar{x}^\nu_{ij}) \ \Tr{m^{-1}_{il}(\mathbf{w}^{\mu'})m_{ikl}[m_{ik}^{-1}(m_{ijk}[m^{-1}_{ij}(\mathbf{x}^\nu),\mathbf{y}^\mu]),\mathbf{z}^{\nu'}]} \\
\nonumber & \ \ \ \ \ \ +\sum_{\nu\mu\nu'\mu'}\sum_{ijkl}(\bar{w}^{\mu'}_{kl} \ \bar{z}^{\nu'}_{li} \ \bar{y}^\mu_{jk} \ \bar{x}^\nu_{ij}) \ \Tr{m^{-1}_{kl}(\mathbf{w}^{\mu'})m_{lik}[\mathbf{z}^{\nu'},m_{ik}^{-1}(m_{ijk}[m^{-1}_{ij}(\mathbf{x}^\nu),\mathbf{y}^\mu])]} \\
\nonumber& \ \ \ \ \ \ +\sum_{\nu\mu\nu'\mu'}\sum_{ijkl}(\bar{w}^{\mu'}_{lk} \ \bar{z}^{\nu'}_{jl} \ \bar{y}^\mu_{ki} \ \bar{x}^\nu_{ij}) \ \Tr{m^{-1}_{kl}(\mathbf{w}^{\mu'})m_{kjl}[m_{jk}^{-1}(m_{kij}[\mathbf{y}^\mu,m^{-1}_{ij}(\mathbf{x}^\nu)]),\mathbf{z}^{\nu'}]} \\
& \ \ \ \ \ \ +\sum_{\nu\mu\nu'\mu'}\sum_{ijkl}(\bar{w}^{\mu'}_{jl} \ \bar{z}^{\nu'}_{lk} \ \bar{y}^\mu_{ki} \ \bar{x}^\nu_{ij}) \ \Tr{m^{-1}_{lj}(\mathbf{w}^{\mu'})m_{lkj}[\mathbf{z}^{\nu'},m_{jk}^{-1}(m_{kij}[\mathbf{y}^\mu,m^{-1}_{ij}(\mathbf{x}^\nu)])]}.
\end{align}
The benefit to choosing our basis as a product-sum form should now be apparent since it is possible to carry out the the unitary rotation $\Ub^\dagger (.) \Ub$ on the terms inside the remaining trace explicitly. The traces are feasible to compute since they are all over sets of $2\times 2$ matrices. Then, we are left with a tedious combinatoric problem that involves substituting the matrix elements from each member of the basis $\mathcal{B}_{2N}$ into the function~\eqref{eq:Trace}, and subsequently combining them to get the curvature via~\eqref{eq:scalar2}. Fortunately, we can also exploit some symmetries and orthogonality properties of the basis elements which nullify terms with disjoint indices and reduce the number of contributions to the summation in~\eqref{eq:scalar2}. Further details of this remaining part of the calculation can be found in Appendix C. Summations over $\mathcal{B}_{2N}$ finally lead to the desired result~\eqref{eq:mainresult}. 
\end{proof}
We can see that, unlike the simple classical formula~\eqref{eq:classcurv} the quantum scalar curvature is not constant and instead is a complicated three-point correlated function of the symplectic spectrum of a given point $\Vb\in\Nn$. In any case our formula is convenient because we can compute the curvature using just a finite summation over the symplectic eigenvalues of the Gaussian state, thus giving us access to its geometric structure without the need to diagonalise and sum over the infinite number of eigenvalues of the density operator. It should also be noted that if there are any degeneracies in the symplectic spectrum we can compute the curvature by taking well-defined limits of the functions $\varphi_1[x]$, $\varphi_2[x,y]$ and $\varphi_3[x,y,z]$ in~\eqref{eq:mainresult}.

\section{Connection between curvature and entropy}

\

Petz \cite{petz1994geometry,petz2002covariance} has provided an intriguing information-theoretic interpretation of the scalar curvature of the KMB metric on the full quantum state space $\mathcal{S}(\mathcal{H})$ that we will now explore in the context of Gaussian states. Recall that the geodesic distance $\mathcal{L}(\Vb_0,\Vb)$ defined in~\eqref{eq:length} quantifies the distinguishability between two states $\Vb_0$ and $\Vb$. Thus we can define a geodesic ball 
\begin{align}
    \mathcal{B}_\epsilon(\Vb_0):=\{\Vb\in\Nn : \ \mathcal{L}(\Vb_0,\Vb)<\epsilon \}
\end{align}
containing all faithful standard Gaussian states that can be distinguished from some point $\Vb_0$ with an effort no greater than $\epsilon$ away from it. The size of this region in the asymptotic limit $\epsilon\to 0$ then gives a measure of the statistical uncertainty of the state $\Vb_0$. From differential geometry we may determine the first non-constant term of the volume of $\mathcal{B}_\epsilon(\Vb_0)$ from the scalar curvature at $\Vb_0$,
\begin{align}\label{eq:vol}
\text{Vol}\big(\mathcal{B}_\epsilon(\Vb_0)\big)=C_{2N}\bigg(\epsilon^{2N}-\frac{\text{Scal}(\Vb_0)}{12(N+1)}\epsilon^{2N+2}\bigg)+\mathcal{O}(\epsilon^{2N+3})
\end{align}
where $C_{2N}$ is the volume of the unit ball in $2N$-dimensional Euclidean space \cite{gallot2004riemannian}. Given our definition of the volume, we may interpret $\text{Scal}(\Vb_0)$ as an \textit{average} measure of statistical uncertainty in the state. In other words, greater curvature implies a smaller inference region around $\Vb_0$, making it harder to distinguish between surrounding states and hence larger statistical uncertainty. This would imply that the curvature should behave similarly to the von-Neumann entropy of the state, which for a Gaussian state is given by \cite{holevo1999capacity}
\begin{align}\label{eq:entropy}
    S_{vN}(\Vb):=-\tr{\hat{\rho}(\Vb) \ \text{ln} \ \hat{\rho}(\Vb)}=\sum_{i=1}^N\bigg(\nu_i+\frac{1}{2}\bigg)\text{log}_2\bigg(\nu_i+\frac{1}{2}\bigg)-\bigg(\nu_i-\frac{1}{2}\bigg)\text{log}_2\bigg(\nu_i-\frac{1}{2}\bigg),
\end{align}
Both the entropy and curvature~\eqref{eq:mainresult} are of course symplectic invariants. However, for our interpretation to hold true we should expect that as we increase (decrease) the curvature of a Gaussian state, this should imply an increase (decrease) in entropy, and vice versa. Thus we formalise the following conjecture,
\begin{conj}
For any pair $\Vb,\Vb'\in\Nn$ of faithful standard Gaussian states,  
\begin{align}\label{eq:conj}
    \mathrm{Scal}(\Vb')> \mathrm{Scal}(\Vb)\Leftrightarrow S_{vN}(\Vb')> S_{vN}(\Vb).
\end{align}
\end{conj}
The monotonicity of the KMB scalar curvature under mixing was conjectured by Petz originally for finite-dimensional quantum systems \cite{petz1994geometry,petz2002covariance}. This can be formally stated in the context of majorisation theory as saying that the scalar curvature of the KMB metric on $\mathcal{S}(\mathcal{H})$ is a Schur concave function \cite{alberti1982stochasticity}. Due to the complicated nature of the curvature this conjecture remains to be proven beyond a two-dimensional Hilbert space, though it has been verified numerically in larger finite-dimensional systems \cite{michor2000curvature,dittmann2000curvature,gibilisco2005monotonicity}. In the case of Gaussian states we can investigate whether or not the conjecture holds in infinite-dimensional settings. Using Theorem 2 the statement can be proven for the single mode ($N=1$) case:
\begin{cor}
    Let $\Vb\in\mathcal{M}_1$ be a single mode covariance matrix. Then $\mathrm{Scal}(\Vb')> \mathrm{Scal}(\Vb)\Leftrightarrow S_{vN}(\Vb')> S_{vN}(\Vb)$.
\end{cor}
\begin{proof}
    Denote the single symplectic eigenvalue of $\Vb$ by $\nu$ and assume $1/2<\nu<\infty$. Then the scalar curvature~\eqref{eq:mainresult} is 
\begin{align}
    \text{Scal}(\Vb)=-\frac{\left(2\nu - f(\nu) + 4\nu^2 f(\nu)\right) \left(f(\nu) + 2\nu(-1 + 6\nu f(\nu)\right)}{8\nu^2(-1 + 4\nu^2) f^2(\nu)}, 
\end{align}
Now the entropy increases strictly monotonically on the interval $\nu\in(1/2,\infty)$ as the state becomes more mixed, so the theorem holds if and only if $\text{Scal}(\Vb)$ is strictly monotonic with respect to $\nu$. Calculating the derivative in $\nu$ directly gives
\begin{align}
\nonumber\frac{d}{d\nu}\text{Scal}(\Vb)&=\frac{8v^3 + f(v) \left(-4 \left(v^2 + 8v^4\right) + 2v \left(-1 + 16 \left(v^2 + v^4\right)\right) f(v) + \left(1 - 4v^2\right)^2 f^2(v)\right)}{4v^3 \left(1 - 4v^2\right)^2 f^3(v)} \\
&\geq\frac{-1 + 4v \left(-1 + 2v (1 + 5v)\right)}{4v^3 \left(-1 + 2v\right) (1 + 2v)^2}> 0
\end{align}
where we used $\text{ln}(x)\leq x-1$ and $1/2<\nu<\infty$.
\end{proof}
Reassuringly, for $N=1$ we also see that $\text{Scal}(\Vb)\to  -2$ as $\nu\to\infty$, meaning we recover the classical result for the standard Gaussian distributions~\eqref{eq:classcurv} when the quantum state is highly mixed. In the other regime, as we approach the vacuum state limit $\nu\to 1/2$ we see the curvature diverges, ie. $\text{Scal}(\Vb)\to-\infty$. Geometrically speaking this means the inference region~\eqref{eq:vol} becomes infinitely large, indicating no statistical uncertainty in the state. 

Proving the Petz conjecture for higher numbers of modes is non-trivial due to the complicated structure of~\eqref{eq:mainresult}, but we can explore this numerically. In Figure 1 we plot both the von Neumann entropy and scalar curvature diagrams for the two-mode ($N=2$) standard Gaussian state as a function of the two symplectic eigenvalues. It can be seen that the curvature shares the same concave behaviour as the entropy thus supporting the conjectured inequality~\eqref{eq:conj}. In all cases we observe a negative curvature, reaching the classical limit~\eqref{eq:classcurv} for the high entropy states. An important subtlety should be noted at the boundaries where the state becomes non-faithful (ie. at least one of the eigenvalues reaches $1/2$). In those limits we again observe a negative divergence in the curvature, though in the multimode case this indicates a breakdown in the Petz conjecture. This is because, while the inference region becomes infinitely large, the state may still have non-zero entropy. This emphasises that the KMB metric is only well-defined on the manifold of faithful Gaussian states. 

\begin{figure}
  \centering
  \subfloat{\includegraphics[width=0.425\textwidth]{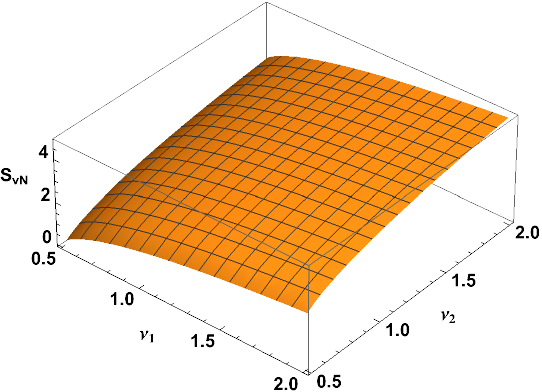}}
  \subfloat{\includegraphics[width=0.425\textwidth]{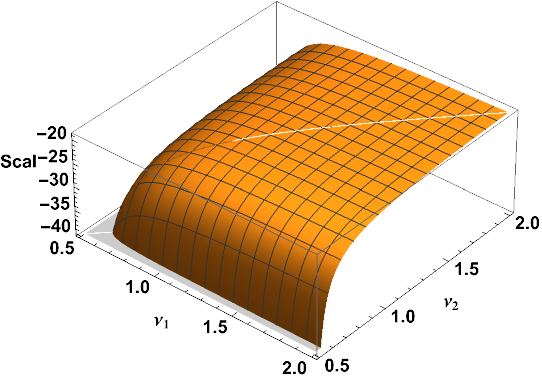}}
  \caption{\textbf{Left}: plot of the von Neumann entropy~\eqref{eq:entropy} for a two-mode standard Gaussian state as a function of the symplectic eigenvalues $\nu_1,\nu_2$. \textbf{Right}: plot of the scalar curvature~\eqref{eq:mainresult} of the same state. Note that regions where $\nu_i\to 1/2$ negatively diverge, while $\text{Scal}(\Vb)\simeq -18$ as $\nu_1\to\infty$ and $\nu_2\to\infty$. \label{sup:grid}}
\end{figure}

For further evidence one should look at much larger systems, and so we next consider a periodic harmonic chain of $N$ modes. Suppose we have a Hamiltonian describing a system of the form \cite{di2020complexity}
\begin{align}
    \hat{H}=\sum^N_{i=1}\frac{1}{2m}\hat{p}_i^2+\frac{m\omega^2}{2}\hat{q}_i^2+\frac{\kappa}{2}(\hat{q}_{i+1}-\hat{q}_{i})^2, \ \ \ \ \  \hat{q}_{N+1}=\hat{q}_{1},
\end{align}
where all modes have equal mass $m$ and frequency $\omega$, while $\kappa$ represents a coupling constant. We assume that both $\kappa$ and $m$ are non-vanishing. A thermal quantum state at inverse temperature $\beta=1/k_B T$ with this Hamiltonian is then a standard Gaussian state of the form 
\begin{align}\label{eq:thermal}
    \hat{\rho}_{eq}(T)=\frac{e^{-\hat{H}/kT}}{Z}=\frac{1}{Z} \ \text{exp}\bigg[-\frac{1}{2}\underline{\hat{x}}^T \Hb_\Vb \underline{\hat{x}}\bigg]
\end{align}
where
\begin{align}\label{eq:ham}
    \Hb_\Vb=\frac{1}{\tilde{T}}\big(\big[(\tilde{\omega}^2+2)\mathbb{I}_N-\mathbf{M}\big]\oplus \mathbb{I}_N\big)
\end{align}
Here we have introduced dimensionless variables $\tilde{\omega}=m\omega^2/\kappa$ and $\tilde{T}=kT/\sqrt{\kappa/m}$. The symmetric matrix $\mathbf{M}$ has non-vanishing elements $M_{i,i+1}=M_{i+1,i}=1$ for $1\leq i\leq N-1$ and $M_{1,N}=M_{N,1}=1$. The covariance matrix can be found by inverting~\eqref{eq:cayley} and then diagonalised using the Williamson decomposition, 
\begin{align}
    \Vb=\frac{1}{2}\text{coth}\big(i\Omb \Hb_{\Vb}/2\big) i\Omb=\Sb_\Vb \big(\Db_\Vb\oplus\Db_\Vb\big)\Sb_\Vb^T
\end{align}
Here $\Db_\Vb=\text{diag}(\nu_1,...\nu_N)$ and the symplectic eigenvalues are
\begin{align}\label{eq:symps}
    \nu_j=\frac{1}{2}\text{coth}\bigg(\frac{\Omega_j}{2\tilde{T}}\bigg), \ \ \ \ j=1,2,...,N
\end{align}
and the effective modes of the system are
\begin{align}
    \Omega_j:=\sqrt{\tilde{\omega}+4\text{sin}^2(\frac{j\pi}{N})}
\end{align}
The symplectic transformation $\Sb_\Vb$ is not needed for our purposes but its form can be found in \cite{di2020complexity}. The curvature can be computed by substituting~\eqref{eq:symps} into our main formula~\eqref{eq:mainresult}, taking care to apply the limits at the degenerate parts of the spectrum.

In the context of thermal states such as~\eqref{eq:thermal}, an alternative way to state the Petz conjecture is to say that the curvature should be an increasing function of temperature $\tilde{T}$ \cite{petz2002covariance}. We can prove this directly for the semi-classical regime when temperature is high $\tilde{T}\gg 1$. First define the ratio between the true curvature and the magnitude of the classical value~\eqref{eq:classcurv}, denoted 
\begin{align}\label{eq:Rcurv}
\mathcal{R}:=\frac{\text{Scal}(\Vb)}{N(2N-1)(N+1)}
\end{align}
A tedious Taylor expansion of~\eqref{eq:mainresult} with respect to $1/\tilde{T}$ yields the leading order quantum correction to the classical curvature, 
\begin{align}
    \mathcal{R}\simeq -1-\frac{\bigg(\frac{1}{2}\sum_{i=1}^N \Omega_i^2+\frac{5}{4}\sum^N_{i<j}\big(\Omega_i^2+\Omega_j^2\big)+\frac{2}{3}\sum^N_{i<j<k}\big(\Omega_i^2+\Omega_j^2+\Omega_k^2\big)\bigg)}{\tilde{T}^2 N(2N-1)(N+1)}, \ \ \ \ \ \ \ \tilde{T}\gg 1.
\end{align}
Clearly the second term in brackets is non-negative regardless of the number of modes, and hence the curvature grows quadratically with increasing temperature in this regime. To check lower temperatures, in Figure 2 (right) we plot $\mathcal{R}$ as a function of $\tilde{T}$ for various values of $\tilde{\omega}$, demonstrating clear monotonicity as well as a negative divergence towards the zero temperature limit. One can see that increasing the frequency causes the dramatic decrease in curvature to begin at higher temperatures, since the system is brought closer to the vacuum state. In Figure 2 (left) $\mathcal{R}$ is plotted as a function of chain length $N$ for a lower temperature regime $(\tilde{\omega}=1, \tilde{T}=0.5)$. It is seen that it approaches a thermodynamic limit with increasing $N$, confirming that the monotoncity is robust for arbitrarily large systems. This provides strong evidence that the conjecture should hold true regardless of the system size.  

\begin{figure}
  \centering
  \subfloat{\includegraphics[width=0.425\textwidth]{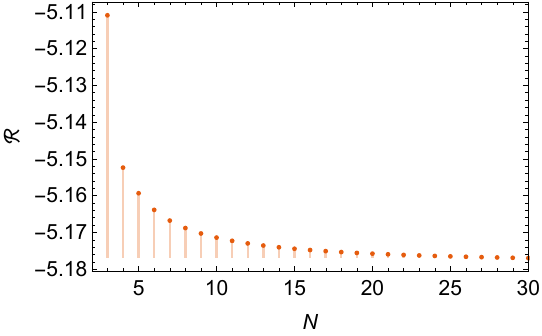}}
  \subfloat{\includegraphics[width=0.425\textwidth]{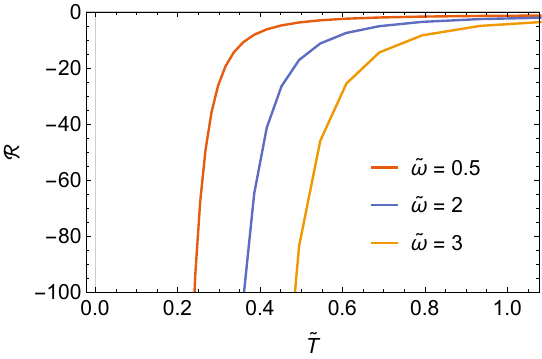}}
  \caption{\textbf{Left}: Plot of relative curvature  $\mathcal{R}$, defined by~\eqref{eq:Rcurv}, as a function of the chain length $N$ for $\tilde{\omega}=1$, $\tilde{T}=0.5$. \textbf{Right}: plot of $\mathcal{R}$ as a function of temperature $\tilde{T}$ for a chain of length $N=50$. \label{sup:grid2}}
\end{figure}

\section{Conclusion}

\

In this paper we have explored the information geometry of the manifold of faithful, standard Gaussian states with respect to the KMB metric. The central geometric objects such as the metric tensor~\eqref{eq:ds}, the geodesic curves~\eqref{eq:length}, and the Riemann~\eqref{eq:riemann} and Ricci~\eqref{eq:ric} curvature tensors were derived. The symplectic invariance of the metric was proven (Theorem 1) and used to derive an analytic expression for the Ricci scalar curvature in terms of the symplectic spectrum of the state covariance matrix (Theorem 2). We then formulated a continuous-variable extension of the Petz conjecture, proposing that the scalar curvature can be used as an entropic measure of statistical uncertainty for a faithful Gaussian state. The conjecture was proven for all single mode states (Corollary 1) and numerically verified for the two-mode state. An $N$-mode periodic thermal chain was also analysed, with the Petz conjecture proven in the semi-classical regime and numerically verified at low temperatures. Our results provide insight into the richer structure of the Gaussian quantum state space in comparison to the space of Gaussian classical distributions, as we observe a breakdown in scale invariance and more negative curvature due to quantum fluctuations. The most dramatic effect was seen at the boundary of non-faithful Gaussian states where we found that the curvature negatively diverges. Technical extensions that could be pursued would be to generalise our results to the full manifold of Gaussian states including those with non-zero displacement. Secondly, it would of course be desirable to have a full proof of the conjecture~\eqref{eq:conj} for all arbitrary $N$-mode states. These results open up a number of promising directions for future applications. Firstly, the quantities we have derived provide a means for investigating the thermodynamic geometry of bosonic systems, including the role of curvature in critical phenomena \cite{ruppeiner1998riemannian,brody2008information,banchi2014quantum} and thermodynamically optimal processes \cite{abiuso2020geometric,mehboudi2022thermodynamic}. The entropic nature of the scalar curvature suggests that it could serve as a geometric tool in continuous-variable majorisation theory \cite{van2023continuous}. As a continuous-variable theory, our formalism could also be useful for investigating the information geometry and Hessian manifolds of quantum fields \cite{erdmenger2020information,kluth2023renormalization}. Another direction would be to understand connections between the KMB metric and measures of mixed Gaussian state non-classicality \cite{yadin2018operational}, quantum correlations \cite{adesso2010quantum} and circuit complexity \cite{di2020complexity}. Finally, the curvature invariants that we have derived could be used to explore the underlying topology of the Gaussian state space, in analogy to studies of ground state manifolds such as \cite{kolodrubetz2013classifying}.

\

\noindent \emph{Acknowledgements. } H. M. acknowledges funding from a Royal Society Research Fellowship (URF/R1/231394). 

\bibliographystyle{apsrev4-1}
\bibliography{thebib.bib}

\appendix

\

\section{A. Integral representation for inverse hyperbolic cotangent matrix function}\label{app:A}

\

\noindent Consider the function 
\begin{align}
    \text{arccoth}(x)=\frac{1}{2}\text{ln}\bigg(\frac{x+1}{x-1}\bigg)
\end{align}
with domain $(-\infty,-1)\cup(1,\infty)$. We note the logarithm can be expressed in integral form:
\begin{align}
    \text{ln} (x)=\int^\infty_0 d\lambda \ \bigg(\frac{1}{1+\lambda}-\frac{1}{x+\lambda}\bigg).
\end{align}
This means we can also express $\text{arccoth}(x)$ as
\begin{align}\label{eq:arcoth}
    \nonumber\text{arccoth}(x)&=\frac{1}{2}\int^\infty_0 d\lambda \bigg(\frac{1}{x+(\lambda-1)}-\frac{1}{x+(\lambda+1)}\bigg), \\
    \nonumber&=\frac{1}{2}\int^\infty_{-1} d\lambda \frac{1}{x+\lambda}-\frac{1}{2}\int^\infty_{1} d\lambda \ \frac{1}{x+\lambda}, \\
    &=\frac{1}{2}\int^1_{-1} d\lambda \ \frac{1}{x+\lambda},
\end{align}
where in the second line performed the change of variables $\lambda\to\lambda-1$ and $\lambda\to\lambda+1$ to the first and second term respectively. It follows that since the matrix $\Vb(i\Omb)$ is hermitian it is diagonalisable, and we can use~\eqref{eq:arcoth} to write
\begin{align}
    \text{arccoth}(2i\Vb\Omb)=\frac{1}{2}\int^1_{-1} d\lambda \ \big[2\Vb (i\Omb)+\lambda\id\big]^{-1}, 
\end{align}
which is well defined for $\Vb>i\Omb/2$ and thus $2\Vb+i\Omb>0$ by transpose. We note the identity
\begin{align}\label{eq:inv_om}
    (i\Omb)\big[2\Vb (i\Omb)+\lambda\id\big]^{-1}=\big[2\Vb+\lambda(i\Omb)\big]^{-1},
\end{align}
which concludes the derivation of~\eqref{eq:int_rep}

\section{B. Diagonalising the covariance matrix }\label{app:B}

\

For calculation purposes we need to express functions such as $[2\Vb+i\lambda\Omb]^{-1}$ in a diagonal form. One can show that the matrix $\Vb(i\Omb)$ is diagonalised using the symplectic decompositions~\eqref{eq:symp} according to \cite{wilde2017gaussian}
\begin{align}\label{eq:symp2}
    \Vb(i\Omb)=\Sb_\Vb\big(\Ub\otimes \id_N\big)\big((-\Db_\Vb)\oplus \Db_\Vb\big)\big(\Ub^\dagger\otimes \id_N\big)\Sb_\Vb^{-1}
\end{align}
where $\id_N$ is the $N\times N$ identity matrix, $\Ub$ is a unitary matrix given by
\begin{align}
    \Ub=\frac{1}{\sqrt{2}}\left(\begin{array}{cc}
          1 & 1   \\ 
          i & -i  
    \end{array}\right)
\end{align}
and $\Db_\Vb=\text{diag}(\nu_1,...\nu_N)$ is the diagonal matrix of symplectic eigenvalues. It is useful to note that $i\Omb=-\sigma_{Y}\otimes\id_N$, with $\sigma_Y$ the Pauli $Y$-matrix, which means we can write
\begin{align}
    \big(\Ub^\dagger \otimes \id_N\big)(i\Omb)=(-\sigma_Z\otimes\id_N)\big(\Ub^\dagger\otimes\id_N\big).
\end{align}
Another useful relation stems from inverting $\Sb\Omb\Sb^T=\Omb$, which gives
\begin{align}
    \Sb^{-1}(i\Omb)=(i\Omb)\Sb^T.
\end{align}
Combining these with~\eqref{eq:symp2} we have
\begin{align}
    \nonumber [2\Vb+i\lambda\Omb]&=\big[2\Vb(i\Omb)+\lambda\big](i\Omb), \\
    \nonumber&=\Sb_\Vb\big(\Ub\otimes \id_N\big)\big((\lambda-2\Db_\Vb)\oplus (\lambda+2\Db_\Vb)\big)\big(\Ub^\dagger\otimes \id_N\big)\Sb_\Vb^{-1}(i\Omb), \\
    \nonumber&=\Sb_\Vb\big(\Ub\otimes \id_N\big)\big((\lambda-2\Db_\Vb)\oplus (\lambda+2\Db_\Vb)\big)\big(\Ub^\dagger\otimes \id_N\big)(i\Omb)\Sb_\Vb^{T}, \\
    \nonumber&=\Sb_\Vb\big(\Ub\otimes \id_N\big)\big((\lambda-2\Db_\Vb)\oplus (\lambda+2\Db_\Vb)\big)\big(-\sigma_Z\otimes \id_N\big)\big(\Ub^\dagger\otimes \id_N\big)\Sb_\Vb^{T}, \\
    &=\Sb_\Vb\big(\Ub\otimes \id_N\big)\big((2\Db_\Vb-\lambda)\oplus (2\Db_\Vb+\lambda)\big)\big(\Ub^\dagger\otimes \id_N\big)\Sb_\Vb^{T},
\end{align}
which, upon inverting, gives us
\begin{align}
    \nonumber[2\Vb+i\lambda\Omb]^{-1}&=(\Sb^T_\Vb)^{-1}\big(\Ub\otimes \id_N\big)\big((2\Db_\Vb-\lambda)^{-1}\oplus (2\Db_\Vb+\lambda)^{-1}\big)\big(\Ub^\dagger\otimes \id_N\big)\Sb_\Vb^{-1}, \\
    &=\Omb\Sb_\Vb\Omb\big(\Ub\otimes \id_N\big)\big((2\Db_\Vb-\lambda)^{-1}\oplus (2\Db_\Vb+\lambda)^{-1}\big)\big(\Ub^\dagger\otimes \id_N\big)\Omb \Sb^T \Omb,
\end{align}
where we used $\Omb\Sb=(\Sb^T)^{-1}\Omb$.

\section{C. Evaluation of trace functionals}\label{app:C}

\

\noindent Consider two matrices of the product-sum form
\begin{align}
    \tilde{\Ub}^\dagger\Ab\tilde{\Ub}=\sum_{\nu=1}^2\sum_{ij}\bar{a}_{ij}^\nu\mathbf{a}^\nu \otimes \eb_{ij}, \ \ \ \ \ \ \tilde{\Ub}^\dagger\Bb\tilde{\Ub}=\sum_{\nu=1}^2\sum_{ij}\bar{b}_{ij}^\nu\mathbf{b}^\nu \otimes \eb_{ij},
\end{align}
By combining~\eqref{eq:Trace} with~\eqref{eq:scalar_HS} we focus on evaluating
\begin{align}
\nonumber K(\Ab,\Bb)=&\frac{1}{4}\sum^2_{\nu\mu\nu'\mu'=1}\sum_{ijkl}\bigg((\bar{b}^{\mu'}_{li} \ \bar{a}^{\nu'}_{kl} \ \bar{b}^\mu_{jk} \ \bar{a}^\nu_{ij}) \ \Tr{m^{-1}_{il}(\mathbf{b}^{\mu'})m_{ikl}[m_{ik}^{-1}(m_{ijk}[m^{-1}_{ij}(\mathbf{a}^\nu),\mathbf{b}^\mu]),\mathbf{a}^{\nu'}]} \\
\nonumber & \ \ \ \ \ \ \ \ \ \ \ \ \  +(\bar{b}^{\mu'}_{kl} \ \bar{a}^{\nu'}_{li} \ \bar{b}^\mu_{jk} \ \bar{a}^\nu_{ij}) \ \Tr{m^{-1}_{kl}(\mathbf{b}^{\mu'})m_{lik}[\mathbf{a}^{\nu'},m_{ik}^{-1}(m_{ijk}[m^{-1}_{ij}(\mathbf{a}^\nu),\mathbf{b}^\mu])]} \\
\nonumber& \ \ \ \ \ \ \ \ \ \ \ \ \ \ +(\bar{b}^{\mu'}_{lk} \ \bar{a}^{\nu'}_{jl} \ \bar{b}^\mu_{ki} \ \bar{a}^\nu_{ij}) \ \Tr{m^{-1}_{kl}(\mathbf{b}^{\mu'})m_{kjl}[m_{jk}^{-1}(m_{kij}[\mathbf{b}^\mu,m^{-1}_{ij}(\mathbf{a}^\nu)]),\mathbf{a}^{\nu'}]} \\
\nonumber& \ \ \ \ \ \ \ \ \ \ \ \ \ +(\bar{b}^{\mu'}_{jl} \ \bar{a}^{\nu'}_{lk} \ \bar{b}^\mu_{ki} \ \bar{a}^\nu_{ij}) \ \Tr{m^{-1}_{lj}(\mathbf{b}^{\mu'})m_{lkj}[\mathbf{a}^{\nu'},m_{jk}^{-1}(m_{kij}[\mathbf{b}^\mu,m^{-1}_{ij}(\mathbf{a}^\nu)])]} \\
\nonumber& \ \ \ \ \ \ \ \ \ \ \ \ \ \ -(\bar{b}^{\mu'}_{li} \ \bar{b}^{\nu'}_{kl} \ \bar{a}^\mu_{jk} \ \bar{a}^\nu_{ij}) \ \Tr{m^{-1}_{il}(\mathbf{b}^{\mu'})m_{ikl}[m_{ik}^{-1}(m_{ijk}[m^{-1}_{ij}(\mathbf{a}^\nu),\mathbf{a}^\mu]),\mathbf{b}^{\nu'}]} \\
\nonumber & \ \ \ \ \ \ \ \ \ \ \ \ \ -(\bar{b}^{\mu'}_{kl} \ \bar{b}^{\nu'}_{li} \ \bar{a}^\mu_{jk} \ \bar{a}^\nu_{ij}) \ \Tr{m^{-1}_{kl}(\mathbf{b}^{\mu'})m_{lik}[\mathbf{b}^{\nu'},m_{ik}^{-1}(m_{ijk}[m^{-1}_{ij}(\mathbf{a}^\nu),\mathbf{a}^\mu])]} \\
\nonumber& \ \ \ \ \ \ \ \ \ \ \ \ \ \ -(\bar{b}^{\mu'}_{lk} \ \bar{b}^{\nu'}_{jl} \ \bar{a}^\mu_{ki} \ \bar{a}^\nu_{ij}) \ \Tr{m^{-1}_{kl}(\mathbf{b}^{\mu'})m_{kjl}[m_{jk}^{-1}(m_{kij}[\mathbf{a}^\mu,m^{-1}_{ij}(\mathbf{a}^\nu)]),\mathbf{b}^{\nu'}]} \\
\nonumber& \ \ \ \ \ \ \ \ \ \ \ \ \ -(\bar{b}^{\mu'}_{jl} \ \bar{b}^{\nu'}_{lk} \ \bar{a}^\mu_{ki} \ \bar{a}^\nu_{ij}) \ \Tr{m^{-1}_{lj}(\mathbf{b}^{\mu'})m_{lkj}[\mathbf{b}^{\nu'},m_{jk}^{-1}(m_{kij}[\mathbf{a}^\mu,m^{-1}_{ij}(\mathbf{a}^\nu)])]}\bigg), 
\end{align}
The traces just consist of $2\times 2$ matrix multiplication and so can be computed in a tedious but straightforward manner;  we give the results line by line:

\

\noindent \textbf{First Line:}
\begin{align*}
\frac{\bar{b}^{\mu'}_{li} \ \bar{a}^{\nu'}_{kl} \ \bar{b}^\mu_{jk} \ \bar{a}^\nu_{ij}}{f_{ij} f_{ik} f_{il} g_{ij} g_{ik} g_{il}} \Biggl( &  B_{kil} f_{il} f_{ik} \Bigl( (a_{11}^{\nu'} a_{21}^\nu b_{11}^\mu b_{12}^{\mu'} + a_{12}^\nu a_{22}^{\nu'} b_{21}^{\mu'} b_{22}^\mu) B_{jik} f_{ij}+ (a_{11}^\nu a_{22}^{\nu'} b_{12}^\mu b_{21}^{\mu'} + a_{11}^{\nu'} a_{22}^\nu b_{12}^{\mu'} b_{21}^\mu) B_{ikj} g_{ij} \Bigr) \\
& \ \  + B_{ilk}f_{il}g_{ik} \Bigl( (a_{21}^{\nu'} a_{21}^\nu b_{12}^{\mu'} b_{12}^\mu + a_{12}^{\nu'} a_{12}^\nu b_{21}^{\mu'} b_{21}^\mu) B_{ijk} f_{ij} + (a_{11}^\nu a_{12}^{\nu'} b_{11}^\mu b_{21}^{\mu'} + a_{21}^{\nu'} a_{22}^\nu b_{12}^{\mu'} b_{22}^\mu) A_{ijk} g_{ij} \Bigr)   \\
+ &  B_{ikl} g_{il} f_{ik} \Bigl( (a_{12}^{\nu'} a_{21}^\nu b_{11}^\mu b_{22}^{\mu'} + a_{12}^\nu a_{21}^{\nu'} b_{11}^{\mu'} b_{22}^\mu) B_{jik} f_{ij}  + (a_{11}^\nu a_{21}^{\nu'} b_{11}^{\mu'} b_{12}^\mu + a_{12}^{\nu'} a_{22}^\nu b_{21}^\mu b_{22}^{\mu'}) B_{ikj} g_{ij} \Bigr) \\ 
&+ A_{ikl}g_{il}g_{ik} \Bigl( (a_{11}^{\nu'} a_{12}^\nu b_{11}^{\mu'} b_{21}^\mu + a_{21}^\nu a_{22}^{\nu'} b_{12}^\mu b_{22}^{\mu'}) B_{ijk} f_{ij} + (a_{11}^{\nu'} a_{11}^\nu b_{11}^{\mu'} b_{11}^\mu + a_{22}^{\nu'} a_{22}^\nu b_{22}^{\mu'} b_{22}^\mu) A_{ijk} g_{ij} \Bigr)   \Biggr)
\end{align*}
\textbf{Second Line:}
\begin{align*}
\frac{\bar{b}^{\mu'}_{kl} \ \bar{a}^{\nu'}_{li} \ \bar{b}^\mu_{jk} \ \bar{a}^\nu_{ij}}{f_{ij} f_{ik} f_{kl} g_{ij} g_{ik} g_{kl}} \Biggl( & f_{kl} B_{lki} f_{ik} \Bigl( (a_{21}^\nu a_{22}^{\nu'} b_{11}^\mu b_{12}^{\mu'} + a_{11}^{\nu'} a_{12}^\nu b_{21}^{\mu'} b_{22}^\mu) B_{jik} f_{ij} + (a_{11}^{\nu'} a_{11}^\nu b_{12}^\mu b_{21}^{\mu'} + a_{22}^{\nu'} a_{22}^\nu b_{12}^{\mu'} b_{21}^\mu) B_{ikj} g_{ij} \Bigr) \\
& + B_{ilk} f_{kl} g_{ik} \Bigl( (a_{12}^{\nu'} a_{21}^\nu b_{12}^\mu b_{21}^{\mu'} + a_{12}^\nu a_{21}^{\nu'} b_{12}^{\mu'} b_{21}^\mu) B_{ijk} f_{ij} + (a_{11}^\nu a_{21}^{\nu'} b_{11}^\mu b_{12}^{\mu'} + a_{12}^{\nu'} a_{22}^\nu b_{21}^{\mu'} b_{22}^\mu) A_{ijk} g_{ij} \Bigr) \\
+ & B_{lik} g_{kl} f_{ik} \Bigl( (a_{12}^{\nu'} a_{21}^\nu b_{11}^{\mu'} b_{11}^\mu + a_{12}^\nu a_{21}^{\nu'} b_{22}^{\mu'} b_{22}^\mu) B_{jik} f_{ij} + (a_{12}^{\nu'} a_{22}^\nu b_{11}^{\mu'} b_{21}^\mu + a_{11}^\nu a_{21}^{\nu'} b_{12}^\mu b_{22}^{\mu'}) B_{ikj} g_{ij} \Bigr) \\
& + A_{lik} g_{kl} g_{ik} \Bigl( (a_{11}^{\nu'} a_{12}^\nu b_{11}^{\mu'} b_{21}^\mu + a_{21}^\nu a_{22}^{\nu'} b_{12}^\mu b_{22}^{\mu'}) B_{ijk} f_{ij} + (a_{11}^{\nu'} a_{11}^\nu b_{11}^{\mu'} b_{11}^\mu + a_{22}^{\nu'} a_{22}^\nu b_{22}^{\mu'} b_{22}^\mu) A_{ijk} g_{ij} \Bigr) \Biggr).
\end{align*}
\textbf{Third Line:}
\begin{align*}
    \frac{\bar{b}^{\mu'}_{lk} \ \bar{a}^{\nu'}_{jl} \ \bar{b}^\mu_{ki} \ \bar{a}^\nu_{ij}}{f_{ij} f_{jk} f_{kl} g_{ij} g_{jk} g_{kl}} \Biggl( & f_{kl} B_{jkl} f_{jk} \Bigl( (a_{12}^\nu a_{22}^{\nu'} b_{11}^\mu b_{21}^{\mu'} + a_{11}^{\nu'} a_{21}^\nu b_{12}^{\mu'} b_{22}^\mu) B_{kji} f_{ij} + (a_{22}^{\nu'} a_{22}^\nu b_{12}^\mu b_{21}^{\mu'} + a_{11}^{\nu'} a_{11}^\nu b_{12}^{\mu'} b_{21}^\mu) B_{ikj} g_{ij} \Bigr) \\
& + B_{klj} f_{kl} g_{jk} \Bigl( (a_{12}^{\nu'} a_{21}^\nu b_{12}^\mu b_{21}^{\mu'} + a_{12}^\nu a_{21}^{\nu'} b_{12}^{\mu'} b_{21}^\mu) B_{kij} f_{ij} + (a_{11}^\nu a_{12}^{\nu'} b_{11}^\mu b_{21}^{\mu'} + a_{21}^{\nu'} a_{22}^\nu b_{12}^{\mu'} b_{22}^\mu) A_{kij} g_{ij} \Bigr) \\
+ & B_{kjl} g_{kl} f_{jk} \Bigl( (a_{12}^\nu a_{21}^{\nu'} b_{11}^{\mu'} b_{11}^\mu + a_{12}^{\nu'} a_{21}^\nu b_{22}^{\mu'} b_{22}^\mu) B_{kji} f_{ij} + (a_{21}^{\nu'} a_{22}^\nu b_{11}^{\mu'} b_{12}^\mu + a_{11}^\nu a_{12}^{\nu'} b_{21}^\mu b_{22}^{\mu'}) B_{ikj} g_{ij} \Bigr) \\
& + A_{kjl} g_{kl} g_{jk} \Bigl( (a_{11}^{\nu'} a_{21}^\nu b_{11}^{\mu'} b_{12}^\mu + a_{12}^\nu a_{22}^{\nu'} b_{21}^\mu b_{22}^{\mu'}) B_{kij} f_{ij} + (a_{11}^{\nu'} a_{11}^\nu b_{11}^{\mu'} b_{11}^\mu + a_{22}^{\nu'} a_{22}^\nu b_{22}^{\mu'} b_{22}^\mu) A_{kij} g_{ij} \Bigr) \Biggr).
\end{align*}
\textbf{Fourth line:}
\begin{align*}
    \frac{\bar{b}^{\mu'}_{jl} \ \bar{a}^{\nu'}_{lk} \ \bar{b}^\mu_{ki} \ \bar{a}^\nu_{ij}}{f_{ij} f_{jk} f_{lj} g_{ij} g_{jk} g_{lj}} \Biggl( & f_{lj} B_{ljk} f_{jk} \Bigl( (a_{11}^{\nu'} a_{12}^\nu b_{11}^\mu b_{21}^{\mu'} + a_{21}^\nu a_{22}^{\nu'} b_{12}^{\mu'} b_{22}^\mu) B_{kji} f_{ij} + (a_{11}^{\nu'} a_{22}^\nu b_{12}^\mu b_{21}^{\mu'} + a_{11}^\nu a_{22}^{\nu'} b_{12}^{\mu'} b_{21}^\mu) B_{ikj} g_{ij} \Bigr) \\
& + B_{klj} f_{lj} g_{jk} \Bigl( (a_{21}^{\nu'} a_{21}^\nu b_{12}^{\mu'} b_{12}^\mu + a_{12}^{\nu'} a_{12}^\nu b_{21}^{\mu'} b_{21}^\mu) B_{kij} f_{ij} + (a_{11}^\nu a_{21}^{\nu'} b_{11}^\mu b_{12}^{\mu'} + a_{12}^{\nu'} a_{22}^\nu b_{21}^{\mu'} b_{22}^\mu) A_{kij} g_{ij} \Bigr) \\
+ & B_{lkj} g_{lj} f_{jk} \Bigl( (a_{12}^\nu a_{21}^{\nu'} b_{11}^\mu b_{22}^{\mu'} + a_{12}^{\nu'} a_{21}^\nu b_{11}^{\mu'} b_{22}^\mu) B_{kji} f_{ij} + (a_{11}^\nu a_{12}^{\nu'} b_{11}^{\mu'} b_{21}^\mu + a_{21}^{\nu'} a_{22}^\nu b_{12}^\mu b_{22}^{\mu'}) B_{ikj} g_{ij} \Bigr) \\
& + A_{lkj} g_{lj} g_{jk} \Bigl( (a_{11}^{\nu'} a_{21}^\nu b_{11}^{\mu'} b_{12}^\mu + a_{12}^\nu a_{22}^{\nu'} b_{21}^\mu b_{22}^{\mu'}) B_{kij} f_{ij} + (a_{11}^{\nu'} a_{11}^\nu b_{11}^{\mu'} b_{11}^\mu + a_{22}^{\nu'} a_{22}^\nu b_{22}^{\mu'} b_{22}^\mu) A_{kij} g_{ij} \Bigr) \Biggr).
\end{align*}
\textbf{Fifth line:}
\begin{align*}
    \frac{\bar{b}^{\mu'}_{li} \ \bar{b}^{\nu'}_{kl} \ \bar{a}^\mu_{jk} \ \bar{a}^\nu_{ij}}{f_{ij} f_{ik} f_{il} g_{ij} g_{ik} g_{il}} \Biggl( & f_{il} B_{kil} f_{ik} \Bigl( (a_{11}^\mu a_{21}^\nu b_{11}^{\nu'} b_{12}^{\mu'} + a_{12}^\nu a_{22}^\mu b_{21}^{\mu'} b_{22}^{\nu'}) B_{jik} f_{ij} + (a_{21}^\mu a_{22}^\nu b_{11}^{\nu'} b_{12}^{\mu'} + a_{11}^\nu a_{12}^\mu b_{21}^{\mu'} b_{22}^{\nu'}) B_{ikj} g_{ij} \Bigr) \\
& + B_{ilk} f_{il} g_{ik} \Bigl( (a_{12}^\mu a_{21}^\nu b_{12}^{\mu'} b_{21}^{\nu'} + a_{12}^\nu a_{21}^\mu b_{12}^{\nu'} b_{21}^{\mu'}) B_{ijk} f_{ij} + (a_{22}^\mu a_{22}^\nu b_{12}^{\mu'} b_{21}^{\nu'} + a_{11}^\mu a_{11}^\nu b_{12}^{\nu'} b_{21}^{\mu'}) A_{ijk} g_{ij} \Bigr) \\
+ & B_{ikl} g_{il} f_{ik} \Bigl( (a_{12}^\nu a_{22}^\mu b_{11}^{\mu'} b_{21}^{\nu'} + a_{11}^\mu a_{21}^\nu b_{12}^{\nu'} b_{22}^{\mu'}) B_{jik} f_{ij} + (a_{11}^\nu a_{12}^\mu b_{11}^{\mu'} b_{21}^{\nu'} + a_{21}^\mu a_{22}^\nu b_{12}^{\nu'} b_{22}^{\mu'}) B_{ikj} g_{ij} \Bigr) \\
& + A_{ikl} g_{il} g_{ik} \Bigl( (a_{12}^\nu a_{21}^\mu b_{11}^{\nu'} b_{11}^{\mu'} + a_{12}^\mu a_{21}^\nu b_{22}^{\nu'} b_{22}^{\mu'}) B_{ijk} f_{ij} + (a_{11}^\mu a_{11}^\nu b_{11}^{\nu'} b_{11}^{\mu'} + a_{22}^\mu a_{22}^\nu b_{22}^{\nu'} b_{22}^{\mu'}) A_{ijk} g_{ij} \Bigr) \Biggr).
\end{align*}
\textbf{Sixth line:}
\begin{align*}
    \frac{\bar{b}^{\mu'}_{kl} \ \bar{b}^{\nu'}_{li} \ \bar{a}^\mu_{jk} \ \bar{a}^\nu_{ij}}{f_{ij} f_{ik} f_{kl} g_{ij} g_{ik} g_{kl}} \Biggl( & f_{kl} B_{lki} f_{ik} \Bigl( (a_{12}^\nu a_{22}^\mu b_{11}^{\nu'} b_{21}^{\mu'} + a_{11}^\mu a_{21}^\nu b_{12}^{\mu'} b_{22}^{\nu'}) B_{jik} f_{ij} + (a_{11}^\nu a_{12}^\mu b_{11}^{\nu'} b_{21}^{\mu'} + a_{21}^\mu a_{22}^\nu b_{12}^{\mu'} b_{22}^{\nu'}) B_{ikj} g_{ij} \Bigr) \\
& + B_{ilk} f_{kl} g_{ik} \Bigl( (a_{12}^\nu a_{21}^\mu b_{12}^{\mu'} b_{21}^{\nu'} + a_{12}^\mu a_{21}^\nu b_{12}^{\nu'} b_{21}^{\mu'}) B_{ijk} f_{ij} + (a_{11}^\mu a_{11}^\nu b_{12}^{\mu'} b_{21}^{\nu'} + a_{22}^\mu a_{22}^\nu b_{12}^{\nu'} b_{21}^{\mu'}) A_{ijk} g_{ij} \Bigr) \\
+ & B_{lik} g_{kl} f_{ik} \Bigl( (a_{11}^\mu a_{21}^\nu b_{11}^{\mu'} b_{12}^{\nu'} + a_{12}^\nu a_{22}^\mu b_{21}^{\nu'} b_{22}^{\mu'}) B_{jik} f_{ij} + (a_{21}^\mu a_{22}^\nu b_{11}^{\mu'} b_{12}^{\nu'} + a_{11}^\nu a_{12}^\mu b_{21}^{\nu'} b_{22}^{\mu'}) B_{ikj} g_{ij} \Bigr) \\
& + A_{lik} g_{kl} g_{ik} \Bigl( (a_{12}^\nu a_{21}^\mu b_{11}^{\nu'} b_{11}^{\mu'} + a_{12}^\mu a_{21}^\nu b_{22}^{\nu'} b_{22}^{\mu'}) B_{ijk} f_{ij} + (a_{11}^\mu a_{11}^\nu b_{11}^{\nu'} b_{11}^{\mu'} + a_{22}^\mu a_{22}^\nu b_{22}^{\nu'} b_{22}^{\mu'}) A_{ijk} g_{ij} \Bigr) \Biggr).
\end{align*}
\textbf{Seventh line:}
\begin{align*}
    \frac{\bar{b}^{\mu'}_{lk} \ \bar{b}^{\nu'}_{jl} \ \bar{a}^\mu_{ki} \ \bar{a}^\nu_{ij}}{f_{ij} f_{jk} f_{kl} g_{ij} g_{jk} g_{kl}} \Biggl( & f_{kl} B_{jkl} f_{jk} \Bigl( (a_{21}^\nu a_{22}^\mu b_{11}^{\nu'} b_{12}^{\mu'} + a_{11}^\mu a_{12}^\nu b_{21}^{\mu'} b_{22}^{\nu'}) B_{kji} f_{ij} + (a_{11}^\nu a_{21}^\mu b_{11}^{\nu'} b_{12}^{\mu'} + a_{12}^\mu a_{22}^\nu b_{21}^{\mu'} b_{22}^{\nu'}) B_{ikj} g_{ij} \Bigr) \\
& + B_{klj} f_{kl} g_{jk} \Bigl( (a_{12}^\nu a_{21}^\mu b_{12}^{\mu'} b_{21}^{\nu'} + a_{12}^\mu a_{21}^\nu b_{12}^{\nu'} b_{21}^{\mu'}) B_{kij} f_{ij} + (a_{22}^\mu a_{22}^\nu b_{12}^{\mu'} b_{21}^{\nu'} + a_{11}^\mu a_{11}^\nu b_{12}^{\nu'} b_{21}^{\mu'}) A_{kij} g_{ij} \Bigr) \\
+ & B_{kjl} g_{kl} f_{jk} \Bigl( (a_{11}^\mu a_{12}^\nu b_{11}^{\mu'} b_{21}^{\nu'} + a_{21}^\nu a_{22}^\mu b_{12}^{\nu'} b_{22}^{\mu'}) B_{kji} f_{ij} + (a_{12}^\mu a_{22}^\nu b_{11}^{\mu'} b_{21}^{\nu'} + a_{11}^\nu a_{21}^\mu b_{12}^{\nu'} b_{22}^{\mu'}) B_{ikj} g_{ij} \Bigr) \\
& + A_{kjl} g_{kl} g_{jk} \Bigl( (a_{12}^\mu a_{21}^\nu b_{11}^{\nu'} b_{11}^{\mu'} + a_{12}^\nu a_{21}^\mu b_{22}^{\nu'} b_{22}^{\mu'}) B_{kij} f_{ij} + (a_{11}^\mu a_{11}^\nu b_{11}^{\nu'} b_{11}^{\mu'} + a_{22}^\mu a_{22}^\nu b_{22}^{\nu'} b_{22}^{\mu'}) A_{kij} g_{ij} \Bigr) \Biggr).
\end{align*}
\textbf{Eighth line:}
\begin{align*}
    \frac{\bar{b}^{\mu'}_{jl} \ \bar{b}^{\nu'}_{lk} \ \bar{a}^\mu_{ki} \ \bar{a}^\nu_{ij}}{f_{ij} f_{jk} f_{lj} g_{ij} g_{jk} g_{lj}} \Biggl( & f_{lj} B_{ljk} f_{jk} \Bigl( (a_{11}^\mu a_{12}^\nu b_{11}^{\nu'} b_{21}^{\mu'} + a_{21}^\nu a_{22}^\mu b_{12}^{\mu'} b_{22}^{\nu'}) B_{kji} f_{ij} + (a_{12}^\mu a_{22}^\nu b_{11}^{\nu'} b_{21}^{\mu'} + a_{11}^\nu a_{21}^\mu b_{12}^{\mu'} b_{22}^{\nu'}) B_{ikj} g_{ij} \Bigr) \\
& + B_{klj} f_{lj} g_{jk} \Bigl( (a_{12}^\mu a_{21}^\nu b_{12}^{\mu'} b_{21}^{\nu'} + a_{12}^\nu a_{21}^\mu b_{12}^{\nu'} b_{21}^{\mu'}) B_{kij} f_{ij} + (a_{11}^\mu a_{11}^\nu b_{12}^{\mu'} b_{21}^{\nu'} + a_{22}^\mu a_{22}^\nu b_{12}^{\nu'} b_{21}^{\mu'}) A_{kij} g_{ij} \Bigr) \\
+ & B_{lkj} g_{lj} f_{jk} \Bigl( (a_{21}^\nu a_{22}^\mu b_{11}^{\mu'} b_{12}^{\nu'} + a_{11}^\mu a_{12}^\nu b_{21}^{\nu'} b_{22}^{\mu'}) B_{kji} f_{ij} + (a_{11}^\nu a_{21}^\mu b_{11}^{\mu'} b_{12}^{\nu'} + a_{12}^\mu a_{22}^\nu b_{21}^{\nu'} b_{22}^{\mu'}) B_{ikj} g_{ij} \Bigr) \\
& + A_{lkj} g_{lj} g_{jk} \Bigl( (a_{12}^\mu a_{21}^\nu b_{11}^{\nu'} b_{11}^{\mu'} + a_{12}^\nu a_{21}^\mu b_{22}^{\nu'} b_{22}^{\mu'}) B_{kij} f_{ij} + (a_{11}^\mu a_{11}^\nu b_{11}^{\nu'} b_{11}^{\mu'} + a_{22}^\mu a_{22}^\nu b_{22}^{\nu'} b_{22}^{\mu'}) A_{kij} g_{ij} \Bigr) \Biggr).
\end{align*}
Each member of the basis set~\eqref{eq:basis} remains in a product-sum form after application of the rotation $\tilde{\Ub}$, and so we may ultimately compute $K(\Ab,\Bb)$ using the above expressions by substituting the relevant matrix elements from each basis member  of $\mathcal{B}_{2N}$. A trick we can use here is to restrict our calculation to a $3\times 3$ subspace for each matrix $\{a^\nu_{ij}\}_{i,j=1,2,3}$ and $\{b^\nu_{ij}\}_{i,j=1,2,3}$,  then replace our final indices according to the rule $1\mapsto i$, $2\mapsto j$ and $3\mapsto k$. Overall this is too tedious to do by hand but can be easily done via assistance from a symbolic computing program. 

Now, since we compute the trace functional~\eqref{eq:scalar_HS} by multiplying a set of diagonal matrices, and the functional is quadratic with respect to pairs of orthogonal basis vectors, terms containing disjoint indices such as $\mathcal{K}(\mathbf{a}_n\otimes\eb_{ii},\mathbf{a}_m\otimes\eb_{jj})$ for $i\neq j$  and $\mathcal{K}(\mathbf{a}_n\otimes\bb_{ij},\mathbf{a}_m\otimes\bb_{kl})=0$ for $i\neq j\neq k\neq l$ will vanish. As a consequence of this, we can divide the summation within~\eqref{eq:scalar2} into nine distinct non-zero contributions, 
\begin{align}
    \text{Scal}(\Vb)=\sum_{l=1}^{9} \mathcal{K}_l(\Vb)
\end{align}
where we eventually find
\begin{align}
    \mathcal{K}_1(\Vb):=\sum_{i}\sum_{n\neq m=1}^3  \mathcal{K}(\mathbf{a}_n \otimes\eb_{ii},\mathbf{a}_m\otimes\eb_{ii})=\sum_i \bigg[\frac{B_{iii} (2 A_{iii}g_{ii}-B_{iii} f_{ii})}{f_{ii}^2 \ g_{ii}^2}\bigg]
\end{align}
\begin{align}
    \mathcal{K}_2(\Vb)=\sum_{i<j}\sum_{n\neq m=1}^2  \mathcal{K}(\mathbf{a}_n\otimes\bb_{ij},\mathbf{a}_m\otimes\bb_{ij})=\sum_{i<j}\frac{1}{4f_{ij}g_{ij}}\bigg[\frac{B_{iij}^2}{ g_{ii} }+\frac{B_{ijj}^2}{   g_{jj}}-\frac{A_{iij} B_{iji}}{ f_{ii}  } - \frac{A_{ijj} B_{jij}}{  f_{jj} } 
\bigg]
\end{align}
\begin{align}
    \nonumber\mathcal{K}_3(\Vb)&:=\sum_{i<j}\sum_{n=1}^3\sum_{m=1}^2 \bigg( \mathcal{K}(\mathbf{a}_n\otimes\eb_{ii},\mathbf{a}_m\otimes\bb_{ij})+\mathcal{K}(\mathbf{a}_n\otimes\eb_{jj},\mathbf{a}_m\otimes\bb_{ij}) \\
    \nonumber& \ \ \ \ \ \ \ \ \ \ \ \ \ \ \ \ \ \ \ \ \ \  \ \ \ \ \ \ \ \  +\mathcal{K}(\mathbf{a}_m\otimes\bb_{ij},\mathbf{a}_n\otimes\eb_{ii})+\mathcal{K}(\mathbf{a}_m\otimes\bb_{ij},\mathbf{a}
    _n\otimes\eb_{jj})\bigg), \\
    \nonumber&=\sum_{i<j}\frac{1}{f_{ij}}\bigg[\frac{A_{iij}^2}{4f_{ii}f_{ij}}+\frac{A_{ijj}^2}{4f_{jj}f_{ij}}-\frac{A_{iij}B_{iii}}{g_{ii}f_{ii}}-\frac{A_{ijj}B_{jjj}}{g_{jj}f_{jj}}-\frac{A_{iij}A_{iii}}{f_{ii}^2}-\frac{A_{ijj}A_{jjj}}{f_{jj}^2}\bigg] \\
    & \ \ \ \ \ \ \ \ \ \  +\frac{1}{g_{ij}}\bigg[\frac{B_{iji}^2}{4f_{ii}g_{ij}}+\frac{B_{jij}^2}{4f_{jj}g_{ij}}+\frac{B_{iij}^2}{f_{ij}g_{ii}}+\frac{B_{ijj}^2}{f_{ij}g_{jj}}-\frac{B_{iji}B_{iii}}{g_{ii}f_{ii}}-\frac{B_{jij}B_{jjj}}{g_{jj}f_{jj}}-\frac{B_{iji}A_{iii}}{f_{ii}^2}-\frac{B_{jij}A_{jjj}}{f_{jj}^2}\bigg]
\end{align}

\begin{align}
    \nonumber \mathcal{K}_4(\Vb)&:=\sum_{i<j}\sum_{n=1}^2  \bigg(\mathcal{K}(\mathbf{a}_n\otimes\bb_{ij},\mathbf{g}_{ij})+\mathcal{K}(\mathbf{g}_{ij},\mathbf{a}_n\otimes\bb_{ij}) +\mathcal{K}(\mathbf{a}_n\otimes\bb_{ij},\tilde{\mathbf{g}}_{ij})+\mathcal{K}(\tilde{\mathbf{g}}_{ij},\mathbf{a}_n\otimes\bb_{ij}), \\
    &=\sum_{i<j}\frac{1}{2g_{ij}}\bigg[\frac{B_{iij}^2}{g_{ii}f_{ij}}+\frac{B_{ijj}^2}{g_{jj}f_{ij}}-\frac{B_{iji}^2}{2g_{ij}f_{ii}}-\frac{B_{jij}^2}{2g_{ij}f_{jj}}\bigg]-\frac{1}{2f_{ij}}\bigg[\frac{A_{iij}B_{iji}}{f_{ii}g_{ij}}+\frac{A_{ijj}B_{jij}}{f_{jj}g_{ij}}+\frac{A_{iij}^2}{2f_{ij}f_{ii}}+\frac{A_{ijj}^2}{2f_{ij}f_{jj}}\bigg],
\end{align}
\begin{align}
    \nonumber\mathcal{K}_5(\Vb)&:=\sum_{i<j}\sum_{n=1}^3 \bigg(  \mathcal{K}(\mathbf{a}_n\otimes\eb_{ii},\mathbf{g}_{ij})+\mathcal{K}(\mathbf{a}_n\otimes\eb_{jj},\mathbf{g}_{ij})  +\mathcal{K}(\mathbf{g}_{ij},\mathbf{a}_n\otimes\eb_{ii})+\mathcal{K}(\mathbf{g}_{ij},\mathbf{a}_n\otimes\eb_{jj}) \\
    \nonumber& \ \ \ \ \ \ \ \ \ \ \ \ \ \ \ \ \ \ \ \ \ \  \ \ \ \ \ \ \ \ +\mathcal{K}(\mathbf{a}_n\otimes\eb_{ii},\tilde{\mathbf{g}}_{ij})+\mathcal{K}(\mathbf{a}_n\otimes\eb_{jj},\tilde{\mathbf{g}}_{ij})  +\mathcal{K}(\tilde{\mathbf{g}}_{ij},\mathbf{a}_n\otimes\eb_{ii})+\mathcal{K}(\tilde{\mathbf{g}}_{ij},\mathbf{a}_n\otimes\eb_{jj})\bigg), \\
    \nonumber&=\sum_{i<j}\frac{A_{iij}}{f_{ii}f_{ij}}\bigg[\frac{A_{iij}}{4f_{ij}}+\frac{B_{iii}}{g_{ii}}+\frac{A_{iii}}{2f_{ii}}\bigg]+\frac{A_{ijj}}{f_{jj}f_{ij}}\bigg[\frac{A_{ijj}}{4f_{ij}}+\frac{B_{jjj}}{g_{jj}}+\frac{A_{jjj}}{2f_{jj}}\bigg] \\
    & \ \ \ \ \ \ \ \ \ \ \ \  +\frac{1}{g_{ij}}\bigg[\frac{B_{iij}^2}{g_{ii}f_{ij}}+\frac{B_{ijj}^2}{g_{jj}f_{ij}}+\frac{B_{iji}}{f_{ii}}\bigg(\frac{B_{iji}}{4g_{ij}}-\frac{B_{iii}}{4g_{ii}}-\frac{A_{iii}}{2f_{ii}}\bigg)+\frac{B_{jij}}{f_{jj}}\bigg(\frac{B_{jij}}{g_{ij}}+\frac{B_{jjj}}{g_{jj}}+\frac{A_{jjj}}{2f_{jj}}\bigg)\bigg]
\end{align}
\begin{align}
    \mathcal{K}_6(\Vb)=\sum_{i<j}\mathcal{K}(\mathbf{g}_{ij},\mathbf{g}_{ij})+\mathcal{K}(\mathbf{g}_{ij},\tilde{\mathbf{g}}_{ij})+\mathcal{K}(\tilde{\mathbf{g}}_{ij},\mathbf{g}_{ij})+\mathcal{K}(\tilde{\mathbf{g}}_{ij},\tilde{\mathbf{g}}_{ij})=\sum_{i<j}\frac{1}{4f_{ij}g_{ij}}\bigg[\frac{B_{iij}^2}{g_{ii}}+\frac{B_{ijj}^2}{g_{jj}}-\frac{A_{iij}B_{iji}}{f_{ii}}-\frac{A_{ijj}B_{jij}}{f_{jj}}\bigg],
\end{align}
\begin{align}
\nonumber\mathcal{K}_7(\Vb)&:=\sum_{i<j<k}\sum_{n,m=1}^2 \bigg( \mathcal{K}(\mathbf{a}_n\otimes\bb_{ij},\mathbf{a}_m\otimes\bb_{jk})+\mathcal{K}(\mathbf{a}_n\otimes\bb_{ik},\mathbf{a}_m\otimes\bb_{jk})+\mathcal{K}(\mathbf{a}_n\otimes\bb_{ij},\mathbf{a}_m\otimes\bb_{ik}) \\
    \nonumber& \ \ \ \ \ \ \ \ \ \ \ \ \ \ \ \ \ \ \ \ \ \  \ \ \ \ \ \ \ \ \ \ \ \ \ \ \   \mathcal{K}(\mathbf{a}_m\otimes\bb_{jk},\mathbf{a}_n\otimes\bb_{ij})+\mathcal{K}(\mathbf{a}_m\otimes\bb_{jk},\mathbf{a}_n\otimes\bb_{ik})+\mathcal{K}(\mathbf{a}_n\otimes\bb_{ik},\mathbf{a}_m\otimes\bb_{ij})\bigg), \\
    \nonumber&=\frac{3}{8}\bigg[\frac{A_{ijk}^2}{f_{ij}f_{ik}f_{jk}}+\frac{B_{ikj}^2}{f_{ij}g_{ik}g_{jk}}+\frac{B_{ijk}^2}{f_{ik}g_{jk}g_{ij}}+\frac{B_{jik}^2}{f_{jk}g_{ij}g_{ik}}\bigg]-\frac{1}{4}\bigg[\frac{A_{jjk}}{f_{jj}f_{jk}}+\frac{B_{jkj}}{f_{jj}g_{jk}}\bigg]\bigg[\frac{B_{jij}}{g_{ij}}+\frac{A_{ijj}}{f_{ij}}\bigg] \\
    & \ \ \ \ \ \ \ \ -\frac{1}{4}\bigg[\frac{A_{ijj}}{f_{ii}f_{ij}}+\frac{B_{iji}}{f_{ii}g_{ij}}\bigg]\bigg[\frac{B_{iki}}{g_{ik}}+\frac{A_{iik}}{f_{ik}}\bigg]-\frac{1}{4}\bigg[\frac{A_{ikk}}{f_{kk}f_{ik}}+\frac{B_{kik}}{f_{kk}g_{ik}}\bigg]\bigg[\frac{B_{kjk}}{g_{jk}}+\frac{A_{jkk}}{f_{jk}}\bigg],
\end{align}
\begin{align}
    \nonumber\mathcal{K}_8(\Vb)&:=\sum_{i<j<k}\bigg( \mathcal{K}(\mathbf{g}_{ij},\mathbf{g}_{jk})+\mathcal{K}(\mathbf{g}_{ik},\mathbf{g}_{jk})  +\mathcal{K}(\mathbf{g}_{jk},\mathbf{g}_{ij})+\mathcal{K}(\mathbf{g}_{jk},\mathbf{g}_{ik})+\mathcal{K}(\mathbf{g}_{ij},\mathbf{g}_{ik})+\mathcal{K}(\mathbf{g}_{ik},\mathbf{g}_{ij})  \\
    \nonumber& \ \ \ \ \ \ \ \ \ \ \ \ \ \ \ \ \ \ \ \ \ \ \ \ \ \ \   +\mathcal{K}(\tilde{\mathbf{g}}_{ij},\tilde{\mathbf{g}}_{jk})+\mathcal{K}(\tilde{\mathbf{g}}_{ik},\tilde{\mathbf{g}}_{jk})  +\mathcal{K}(\tilde{\mathbf{g}}_{jk},\tilde{\mathbf{g}}_{ij})+\mathcal{K}(\tilde{\mathbf{g}}_{jk},\tilde{\mathbf{g}}_{ik})++\mathcal{K}(\tilde{\mathbf{g}}_{ij},\tilde{\mathbf{g}}_{ik})+\mathcal{K}(\tilde{\mathbf{g}}_{ik},\tilde{\mathbf{g}}_{ij}) \\
    \nonumber& \ \ \ \ \ \ \ \ \ \ \ \ \ \ \ \ \ \ \ \ \ \ \ \ \ \ \  +\mathcal{K}(\mathbf{g}_{ij},\tilde{\mathbf{g}}_{jk})+\mathcal{K}(\mathbf{g}_{ik},\tilde{\mathbf{g}}_{jk})  +\mathcal{K}(\mathbf{g}_{jk},\tilde{\mathbf{g}}_{ij})+\mathcal{K}(\mathbf{g}_{jk},\tilde{\mathbf{g}}_{ik})+\mathcal{K}(\mathbf{g}_{ij},\tilde{\mathbf{g}}_{ik})+\mathcal{K}(\mathbf{g}_{ik},\tilde{\mathbf{g}}_{ij}) \\
    \nonumber& \ \ \ \ \ \ \ \ \ \ \ \ \ \ \ \ \ \ \ \ \ \ \ \ \ \ \  +\mathcal{K}(\tilde{\mathbf{g}}_{ij},\mathbf{g}_{jk})+\mathcal{K}(\tilde{\mathbf{g}}_{ik},\mathbf{g}_{jk})  +\mathcal{K}(\tilde{\mathbf{g}}_{jk},\mathbf{g}_{ij})+\mathcal{K}(\tilde{\mathbf{g}}_{jk},\mathbf{g}_{ik})+\mathcal{K}(\tilde{\mathbf{g}}_{ij},\mathbf{g}_{ik})+\mathcal{K}(\tilde{\mathbf{g}}_{ik},\mathbf{g}_{ij})\bigg), \\
    &=\mathcal{K}_7(\Vb),
\end{align}
\begin{align}
   \nonumber\mathcal{K}_9(\Vb)&:=\sum_{i<j<k}\sum_{n=1}^2 \bigg( \mathcal{K}(\mathbf{a}_n\otimes\bb_{ij},\mathbf{g}_{jk})+\mathcal{K}(\mathbf{a}_n\otimes\bb_{ik},\mathbf{g}_{jk})  +\mathcal{K}(\mathbf{g}_{jk},\mathbf{a}_n\otimes\bb_{ij})+\mathcal{K}(\mathbf{g}_{jk},\mathbf{a}_n\otimes\bb_{ik}), \\
   \nonumber& \ \ \ \ \ \ \ \ \ \ \ \ \ \ \ \ \ \ \ \ \ \  +\mathcal{K}(\mathbf{a}_n\otimes\bb_{ij},\tilde{\mathbf{g}}_{jk})+\mathcal{K}(\mathbf{a}_n\otimes\bb_{ik},\tilde{\mathbf{g}}_{jk})  +\mathcal{K}(\tilde{\mathbf{g}}_{jk},\mathbf{a}_n\otimes\bb_{ij})+\mathcal{K}(\tilde{\mathbf{g}}_{jk},\mathbf{a}_n\otimes\bb_{ik})\\
   \nonumber& \ \ \ \ \ \ \ \ \ \ \ \ \ \ \ \ \ \ \ \ \ \  +\mathcal{K}(\mathbf{a}_n\otimes\bb_{ij},\mathbf{g}_{ik})+\mathcal{K}(\mathbf{a}_n\otimes\bb_{ik},\mathbf{g}_{ij})  +\mathcal{K}(\mathbf{a}_n\otimes\bb_{ij},\tilde{\mathbf{g}}_{ik})+\mathcal{K}(\mathbf{a}_n\otimes\bb_{ik},\tilde{\mathbf{g}}_{ij})  \bigg), \\
   &=2\mathcal{K}_7(\Vb)
\end{align}
Note that we simplified these expressions using the symmetries $f_{ij}=f_{ji}$, $g_{ij}=g_{ji}$, $B_{ijk}=B_{kji}$ and $A_{ijk}=A_{ikj}=A_{jik}=A_{kij}=A_{kji}$. Grouping these expressions together and simplifying yields the final form for the scalar curvature~\eqref{eq:mainresult}.

\end{document}